\documentclass[11pt,reqno]{amsart}
\usepackage{amsfonts,amssymb,amsmath,amsopn,amsthm,graphicx}
\usepackage{amsxtra, mathrsfs}
\usepackage{colordvi}
\usepackage{euscript}
\usepackage[usenames,dvipsnames]{color}
\usepackage{times} 
\usepackage{amsfonts,amssymb,amsbsy,amsmath,amsthm,dsfont}

 \numberwithin{equation}{section}

\newtheorem{thm}{Theorem}[section]
\newtheorem{cor}[thm]{Corollary}
\newtheorem{assumption}{Assumption}
\newtheorem{lem}[thm]{Lemma}
\newtheorem{prop}[thm]{Proposition}
\newtheorem{rem}[thm]{Remark}

\newcommand{\rig}{\big\rangle}
\newcommand{\lef}{\big\langle}

\newcommand{\B}{\mathscr{B}}

\newcommand{\HH}{{\rm H}}
\newcommand{\eps}{\varepsilon}
\newcommand{\N}{\mathbb{N}}
\newcommand{\R}{\mathbb{R}}

\newcommand{\C}{\mathbb{C}}

\newcommand{\h}{\mathcal{H}}
\newcommand{\hh}{\mathscr{H}}

\newcommand{\U}{{\rm U}}

\newcommand{\pd}{\partial}

\newcommand{\F}{\mathcal{F}}
\newcommand{\E}{\mathcal{E}}
\newcommand{\id}{\mathds{1}}

\title[Impurity-bound excitons]{Impurity-bound excitons in one and two dimensions}

\author{Horia Cornean${}^*$}\thanks{${}^*$ Department of Mathematical Sciences, Aalborg University, Skjernvej 4A, 9220 Aalborg, Denmark} \author{Hynek Kova\v{r}\'{\i}k${}^{**}$}\thanks{${}^{**}$ DICATAM, Sezione di Matematica
Universit{\`a} degli studi di Brescia,
Via Branze, 38-25123
Brescia, Italy}\author{Thomas G. Pedersen${}^{***}$}\thanks{${}^{***}$ Department of Materials and Production, Aalborg University, Skjernvej 4A, 9220 Aalborg, Denmark}

\begin{document}

\begin{abstract}
We study three-body Schr\"odinger operators in one and two dimensions modelling an exciton interacting with a charged impurity. We consider certain classes of  multiplicative interaction potentials proposed in the physics literature. We show that if the impurity charge is larger than some critical value, then  three-body bound states cannot exist. Our spectral results are confirmed by variational numerical computations based on projecting on  a finite dimensional subspace generated by a Gaussian basis.
\end{abstract}

\maketitle

\section{\bf Introduction and main results}\label{sect-1}

Three-body complexes, in which one particle is oppositely charged from the other two, play an important role in solid-state physics. Such complexes are typically encountered when excitons (i.e. two-body complexes consisting of a negative electron and a positive hole) interact with a third charge. If the third particle is an additional mobile electron or hole, charged excitons (trions) may form. Alternatively, excitons interacting with an immobile charged impurity may lead to impurity-bound excitons \cite{Mostaani,Ganchev}. The latter can be modelled as a light electron-hole pair interacting with an infinitely heavy impurity charge $\kappa$. n two previous papers we studied in detail one-dimensional impurity-bound excitons where the interactions were modelled by contact potentials \cite{hkpc}, as well as trions \cite{rpc}. In the current manuscript, we extend the analysis to the physically more relevant case of two-dimensional atomically thin semiconductors, in which impurity-bound excitons are frequently observed. We consider interactions given by multiplicative potential operators of the Keldysh form \cite{Keldysh, cud, tpv}, both in one and two dimensions. Hence, in this paper we study the spectral properties of the operators
\begin{equation}
 \HH_{\kappa,\lambda}(V) = -\Delta -\kappa V(x) +\lambda V(y) - V(x-y)  \quad \text{in} \ \ L^2(\R^{2d}), \quad d=1,2,
\end{equation}
where $V:\R^d\to \R$ is a potential function and $\kappa, \lambda$ are positive coupling constants. In the sequel we will adopt the notation
$$
\HH_\kappa(V):=\HH_{\kappa,\kappa}(V) .
$$
The operator $\HH_\kappa(V)$ describes an impurity of infinite mass interacting with an exciton. The impurity charge $\kappa$ controls how the impurity interacts with the electron and the hole. Our main interest here is to show that if $\kappa$ is larger than some critical value, then generically, $\HH_\kappa(V)$ does not have "three-body" bound states. Also, we numerically analyze the asymptotic behavior for $|\kappa|\ll 1$ and demonstrate important differences with respect to the contact potential model. Our analytical findings are supported by numerical results for both critical and asymptotic limits.

Very roughly said, the generic situation is the following: if $\kappa>0$ is small enough, then one expects at least one discrete eigenvalue (even infinitely many for the class of $2d$ potentials we consider), while when $\kappa$ is larger than some critical value, no discrete eigenvalues can exist.

In the $1d$ case we show in Theorem \ref{thm-strong} and Corollary \ref{cor-1} that if the interaction potential is even, localized, smooth enough and with a non-degenerate maximum at zero, then the above "generic" case applies. 
 Nevertheless, in Proposition \ref{prop-well} we construct a flat-well potential which has at least one bound state for all $\kappa>1$.

In the $2d$ case we only consider the "physical" potential proposed by \cite{cud} (see also \eqref{v-appr}), which has a logarithmic divergence near the origin and goes like $-1/|x|$ at infinity. For this particular potential we show in Theorem \ref{thm-2d} that $\HH_\kappa(V)$ has infinitely many discrete eigenvalues for a certain interval of variation for $\kappa$, but no eigenvalues at all for large enough $\kappa$. 

\subsection{Notation}  
Given a set $M$ and two functions $f_1,\, f_2:M\to\R$, we write $f_1(m) \lesssim f_2(m)$ if there exists a numerical constant $c$ such that $f_1(m) \leq c\, f_2(m)$ for all $m\in M$. The symbol $f_1(m) \gtrsim f_2(m)$ is defined analogously. Moreover, we use the notation 
$$
f_1(m) \asymp f_2(m)  \quad \Longleftrightarrow \quad f_1(m) \lesssim f_2(m) \ \wedge \ f_2(m) \lesssim f_1(m)
$$
Given $f,g\in L^2(\R^n)$ we denote by
$$
\lef f, g \rig_{L^2(\R^n)} = \int_{\R^n} \overline{ f(x)}\cdot g(x)\, dx
$$
the scalar product of $f$ and $g$. Since our operators are real, we will often work with real functions, only, especially when we perform variational arguments. Given a self-adjoint operator $A$ on a Hilbert space $\mathcal H$ 
we will denote by $N(A, \tau)_{\mathcal H}$ the number of eigenvalues of $A$ less than $\tau$ counted with their multiplicities. 

\vspace{0.2cm}

Now we formulate our main results.

\vspace{0.2cm}

\subsection{The one-dimensional case}

\noindent Let $v: \R\to \R$ be a potential function  satisfying the following 
\begin{assumption} \label{ass-v} We have

\begin{enumerate} 

\medskip
\item $v\in C^3_0(\R)$ and $v(x)\geq 0$ for all $x\in\R$.  

\medskip

\item $v(x)=v(-x)$ for all $x\in\R$. 

\medskip

\item $v(x) < v(0)$ for all $x\neq 0$ and $v''(0) < 0$. 

\end{enumerate}  
\end{assumption}

\noindent Given such a $v$, it is easily seen that  the operator  $\HH_{\kappa,\lambda}(v)$ in $L^2(\R^2)$ is associated with the closed quadratic form
\begin{align*} 
q_{\kappa, \lambda} [u] & = \int_{\R^2} |\nabla u|^2\, dx dy - \kappa \int_{\R^2} u^2(x,y) v(x)\, dx dy + \lambda \int_{\R^2} u^2(x,y) v(y)\, dx dy \\
& \quad -  \int_{\R^2} u^2(x,y) v(x-y)\, dx dy ,
\end{align*}
with $ d(q)= H^1(\R^2)$. Let $
q_{\kappa, \kappa} [u]= q_\kappa [u] .
$


\begin{thm} \label{thm-strong}
Let assumption \ref{ass-v} be satisfied. Then to any $\lambda >1$ there exists $\kappa_c(\lambda)>0$ such that 
\begin{equation} 
\kappa \geq \kappa_c(\lambda) \quad \Longrightarrow \quad \sigma_d( \HH_{\kappa,\lambda}(v)) = \emptyset. 
\end{equation}
\end{thm}

\smallskip

\begin{cor} \label{cor-1} 
There exists $\kappa_c$ such that the discrete spectrum of $\HH_\kappa(v)$ is empty for all $\kappa\geq \kappa_c$. 
\end{cor}

\noindent The following example indicates that the non-flatness condition in Assumption \ref{ass-v} cannot be omitted. 

\begin{prop} \label{prop-well}
Let $w: \R \to \R$ be given by 
$$
w(x) =  \begin{cases}
w_0
  & \text{for }\quad |x| \leq 1
  \,,
  \\
  0
  & \text{for }  \, \quad  |x| > 1
  \,.
\end{cases}
$$
Then there exists $w_c \in (0,\infty)$ such that for $w_0 \geq w_c$ the operator
$$
\HH_\kappa(w) = -\pd_x^2 -\pd_y^2 -\kappa\, w(x) +\kappa\, w(y) - w(x-y) 
$$
in $L^2(\R^2)$ has at least one discrete eigenvalue for all $\kappa>1$.
\end{prop}

\vspace{0.2cm}

\subsection{The two-dimensional case.} 

It was shown in \cite{cud} that the Coulomb potential energy created by a point charge at the origin that
electrons feel in a two-dimensional layer is well approximated by the function
\begin{equation} \label{v-appr}
V_{\rm ctr}(x) = {\frac{1}{r_0}} \log\frac{|x|}{|x|+{r_0}} -w(|x|), \qquad x\in\R^2,
\end{equation}
{ where $r_0>0$ is a constant\ } and $w:\R_+\to \R$ satisfies 

\begin{assumption} \label{ass-w} We have

\begin{enumerate} 

\medskip
\item $w\in C^2(\R_+)$ and  for all $r\in\R_+$ it holds
$$
0 \, \leq \, w(r) \, \leq\, w(0)\, .
$$

\medskip

\item {$w(r)= \mathcal{O}(e^{-r})$ as $r\to\infty$. } 

\end{enumerate}  
\end{assumption}

\smallskip 
\noindent {Without loss of generality in what follows we will put
$$
r_0 =1.
$$
}
\noindent Let us consider the operator 
\begin{equation} \label{ham-2d}
\hh_\kappa = -\Delta_x -\Delta_y  +\kappa V_{\rm ctr}(x) -\kappa V_{\rm ctr}(y) + V_{\rm ctr}(|x-y|) , \qquad x,y \in\R^2
\end{equation} 
in $L^2(\R^4)$.

\begin{thm} \label{thm-2d}
The discrete spectrum of $\hh_\kappa$ 

\noindent {\rm (i)} is empty for $\kappa$ larger than some critical value,  

\noindent {\rm (ii)} but contains infinitely many eigenvalues for certain values of $\kappa\in (1/2,1)$.  
\end{thm}

\begin{figure}
\includegraphics[width=0.70\linewidth]{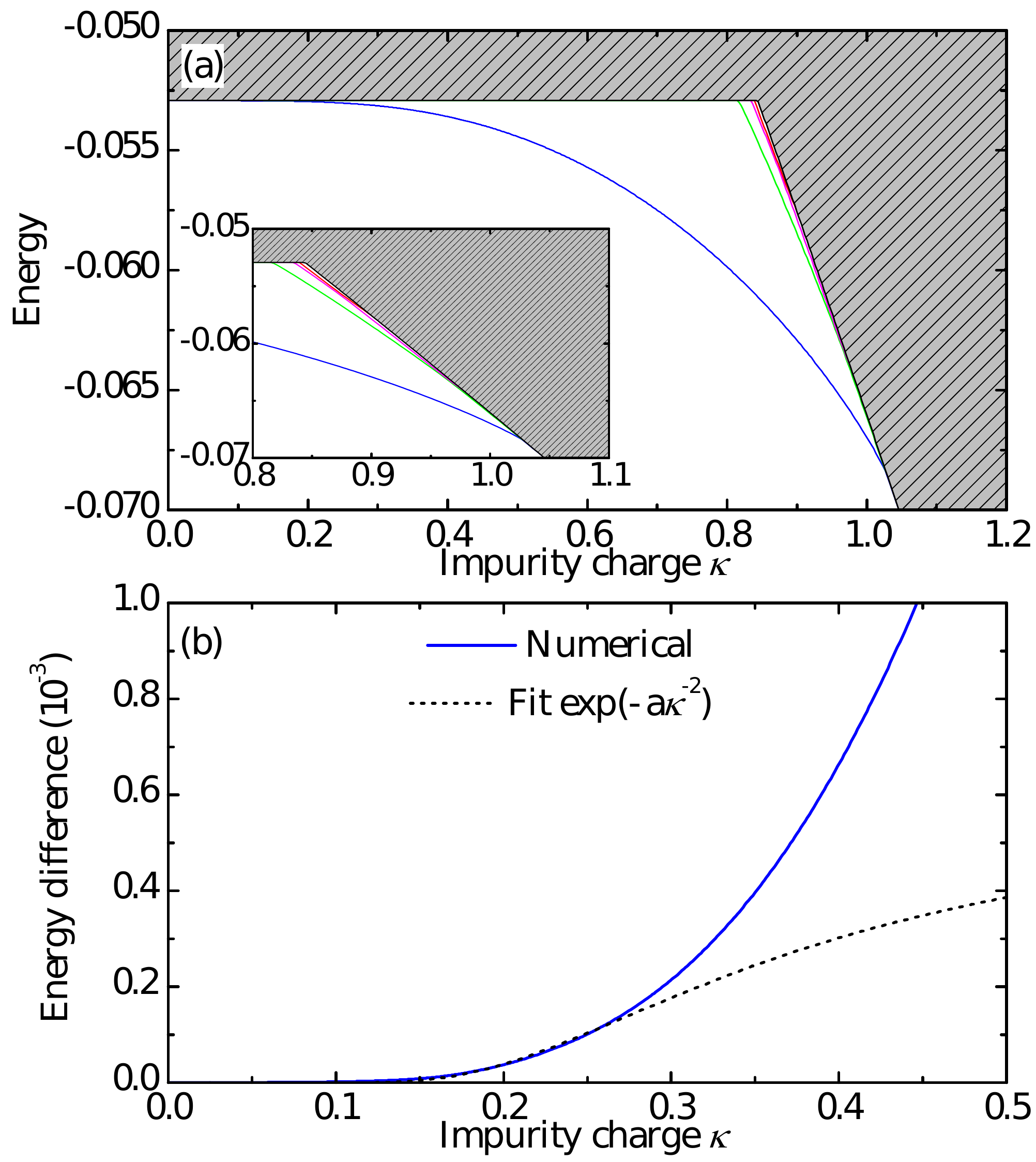}
\caption[Numerics]{\label{fig:1}%
(a) Numerical eigenvalues computed from expansion in a Gaussian basis. The hatched area illustrates the continuous spectrum while the colored lines are discrete eigenvalues. The inset is a zoom near the critical region. (b) Separation of the fundamental discrete eigenvalue from the continuum (blue solid line) and comparison with an analytical fit (black dashed line).}
\end{figure}%

\subsection{Numerical results}
To illustrate and support the exact finding of the present work, we now analyze numerically a concrete model of impurity-bound excitons in $d=2$. To this end, we apply the full Keldysh potential
\begin{equation} \label{v-keldysh}
V(x) = \frac{\pi}{2r_0}\left\{H_0\left(\frac{|x|}{r_0}\right) - Y_0\left(\frac{|x|}{r_0}\right)\right\}, 
\end{equation}
where $H_0$ and $Y_0$ are Struve and Bessel functions, respectively, and we take $r_0=20$. We expand eigenstates in a Gaussian basis 
$\psi_{nmp}=\exp(-\alpha_n x^2-\beta_m y^2-\gamma_p (x-y)^2)$ with exponents between $0$ and $7$. A total of 320 basis functions are used in the expansion. In the unperturbed case, $\kappa=0$, an exciton binding energy  $\Lambda_0(V)\sim -0.0529$ is found (see \eqref{lambda0} for the definition of $\Lambda_0(V)$). This value also gives the lower bound of the essential spectrum for small $\kappa$. In Fig. 1a, the continuum is illustrated by the hatched area. As $\kappa$ is increased above $k_e\sim 0.844$, the bottom of the essential spectrum  is given by the two-body electron-impurity complex instead (see also \eqref{lambda0} for the definition of $\Lambda_1(\kappa,V)$).

In Proposition \ref{prop-es} we show that for the more general class of potentials we consider, the value of $k_e$ always lies in the interval $(1/2,1)$.

As illustrated by the colored lines in Fig. 1a, discrete eigenstates exist when $0<\kappa<\kappa_c\approx 1.029$. Only a single discrete eigenvalue (marked by the blue line) exists in the entire range $0<\kappa<\kappa_c$ with the others only emerging above a certain lower critical value $\tilde{\kappa}_c$, e.g. $\tilde{\kappa}_c\approx 0.815$ for the second eigenvalue (shown in green).

It is particularly interesting to investigate the $\kappa$-dependence of the fundamental discrete eigenvalue shown in blue in Fig. 1a. Hence, in Fig. 1b, we have shown the difference between this state and the bottom of the continuum. It is immediately clear from the plot that this energy difference has a very weak $\kappa$-dependence in the asymptotic limit $\kappa\rightarrow 0$. In the figure, we have fitted the numerical behavior to the analytical form $\Delta E=A \exp(-a\kappa^{-2})$. A rather satisfactory fit is observed for $\kappa\lesssim 0.25$.

The rigorous analysis of the small $\kappa$ behaviour will be done elsewhere, but let us give a hand-waving argument for why one should expect a binding energy which goes like $\exp(-a\kappa^{-2})$. The explanation is that our operator is somehow similar with a one-body $2d$-Schr\"odinger operator with a potential $\kappa W$ where $\int_{\R^2}W(x)dx=0$. Thus the perturbation is effectively of order $\kappa^2$. Up to a Birman-Schwinger argument, and knowing that the resolvent of the free Laplacian in $2d$ has a logarithmic  threshold behavior, one expects to have a bound state $\lambda<0$, $|\lambda|\ll 1$,  which obeys an estimate of the form $\log(-\lambda)\kappa^2\sim -1$, \cite{si2}.

\subsection{The structure of the paper}

After the Introduction, in Section \ref{sect-2} we identify the essential spectrum of this class of operators, a result which is valid for both dimensions. In Section \ref{sect-3} we treat the one-dimensional case, while in Section \ref{sect-4} we deal with $2d$. We end with an Appendix.

\vspace{0.2cm}

\section{\bf Preliminaries}\label{sect-2}

\subsection{The essential spectrum.} 

\begin{prop} \label{prop-es}
Let $V\geq 0$ be non-zero and assume that $V\in L^p(\R^d) +L^\infty_\eps(\R^d)$ with $p=1$ if $d=1,$ and $p>1$ if $d=2$. Then there exists $k_e \in \big (\frac{1}{2}, 1\big )$ such that 
\begin{equation}  \label{lambda-k}
\sigma_{\rm es} (\HH_{\kappa,\lambda}(V)) = [\Lambda(\kappa,V), \infty),  \qquad \forall\, \lambda >0,
\end{equation}
where 
\begin{equation} \label{ke}
\Lambda(\kappa,V) = \begin{cases}
\inf \sigma (-\Delta - V(x-y))  
  & \text{if }\quad \kappa < k_e 
  \,,
  \\
\inf \sigma (-\Delta - \kappa V(x))   
  & \text{if }  \, \quad  \kappa \geq k_e
  \,.
\end{cases}
\end{equation}
\end{prop}

\begin{proof}
Let 
$$
\Lambda_0(V) = \inf \sigma (-\Delta - V(x-y))  , \quad \text{and} \quad \Lambda_1(\kappa,V)  = \inf \sigma (-\Delta - \kappa V(x))  \, .
$$
Since $V\geq 0$, the HVZ-theorem (see e.g.~\cite{si}) implies that \eqref{lambda-k} holds true with 
\begin{equation} \label{eq-hvz}
\Lambda(\kappa,V) =  \min \left\{\Lambda_0(V), \Lambda_1(\kappa,V) \right\}. 
\end{equation}
By introducing the new variables $s= x-y$ and $t= \frac{x+y}{2}$ we find that $-\Delta - V(x-y)$ is unitarily equivalent to the operator 
$
-2 \Delta_s -\frac 12\, \Delta_t -V(s) 
$
in $L^2(\R^{2d})$. Hence if $\kappa>0$ we have
\begin{equation} \label{lambda0}
\Lambda_0(V) = \inf \sigma(-2 \Delta_s  -V(s) )< 0, \quad \Lambda_1(\kappa,V) = \inf \sigma(-\Delta_x  -\kappa V(x) ) < 0,
\end{equation}
where the strict inequalities follow from the fact that $V\geq 0, V\neq 0$ and $d\leq 2$. On the other hand, $\Lambda_1(0,V) =0$ and standard spectral theory arguments show that $\Lambda_1(\cdot,V)$ is a continuously decreasing function of $\kappa$ which obeys $\Lambda_1(\kappa,V) \to -\infty$ as $\kappa\to\infty$. This implies that there exists a unique $k_e>0$ for  which $\Lambda_0(V)=\Lambda_1(k_e,V)$. Now if $\kappa\geq 1$ we have the inequalities 
$$
-2\Delta_s -V(s)> -\Delta_s -V(s)\geq -\Delta_x -\kappa V(x),
$$
thus in view of equations \eqref{eq-hvz} and \eqref{lambda0} we conclude that $k_e<1$.  Also, if $0<\kappa\leq 1/2$ we have 
$$
0>\Lambda_1(\kappa,V)\geq \Lambda_1(2^{-1},V)=2^{-1} \inf \sigma(-2\Delta -V)=2^{-1}\Lambda_0(V)>\Lambda_0(V),
$$
which shows that $k_e>1/2$.
 
\end{proof}

\begin{rem}
The exact value of $k_e$ depends on the profile of $V$. If $d=1$ and $V$ is replaced by a Dirac delta quadratic form, then $k_e=\frac{1}{\sqrt{2}}$, see \cite{hkpc}. 
\end{rem}

\section{\bf  Proofs in the one-dimensional case}\label{sect-3}

\subsection{Auxiliary results}

\noindent To prove Theorem  \ref{thm-strong} we will need several auxiliary results.  Let us introduce a scaling function 
$\U_\kappa: L^2(\R^2) \to L^2(\R^2)$ 
given by 
\begin{equation}€ \label{eq-U}
(\U_\kappa\,  f)(x,y) =  \kappa^{\frac 14} f\big( \kappa^{\frac 14} x,  \kappa^{\frac 14} y)\, .
\end{equation}
Then $\U_\kappa$ maps $L^2(\R^2)$ unitarily onto itself.  We define the operator 
\begin{equation} \label{ham-scaled-def} 
\h_{\kappa,\lambda} = \frac{1}{\sqrt{\kappa}}\, \U^*_\kappa\,  \HH_{\kappa,\lambda}(v)\, \U_\kappa
\end{equation} 
in $L^2(\R^2)$. Obviously 
\begin{equation} \label{equiv}
 \sigma_d( \HH_{\kappa,\lambda}(v)) = \emptyset \quad \Longleftrightarrow \quad  \sigma_d(\h_{\kappa,\lambda}) = \emptyset.
\end{equation} 

\medskip

\noindent Next we define the operator
\begin{equation} \label{h-1d}
h_\kappa = -\pd_x^2 -\sqrt{\kappa}\ v( \kappa^{-\frac 14}\, x) 
\end{equation}
in $L^2(\R)$. Clearly 
\begin{equation} 
\sigma_{es} (h_\kappa) = [0,\infty) \quad \forall\ \kappa >0,
\end{equation}
and a simple calculation shows that 
\begin{equation} \label{ham-scaled} 
\h_{\kappa,\lambda} =  h_\kappa -\pd_y^2 +\frac{\lambda}{\sqrt{\kappa}}\, v(\kappa^{-\frac 14}\, y) - \frac{1}{\sqrt{\kappa}}\ v\left(\kappa^{-\frac 14}\, (x-y)\right) .
\end{equation} 
Similarly as above we use the notation 
$$
\h_{\kappa,\kappa} =\h_\kappa \, .
$$

\vspace{0.2cm}

\noindent Now we turn our attention to the case of large $\kappa$. Let
\begin{equation} \label{beta}
\omega= \sqrt{-\frac{v''(0)}{2}}\, .
\end{equation}
We have

\begin{lem} \label{lem-1}
Let $E_1(\kappa)$ and $E_2(\kappa)>E_1(\kappa)$ be the two lowest eigenvalues of $h_\kappa$. Then 
\begin{equation} \label{gap}
\lim_{\kappa\to\infty} (E_j(\kappa) +\sqrt{\kappa}\ v(0) ) = (2j-1)\, \omega, \qquad j=1,2\, .
\end{equation} 
\end{lem}

\begin{proof}
Let $\chi:\R\to [0,1]$ and $ \widetilde\chi:  \R\to [0,1]$ be two $C_0^\infty$ functions such that
$$
\chi(x) = 1 \ \ \  \forall\ x\in[-1,1], \qquad \chi(x) = 0 \ \ \ \forall\, x: \ |x| >2,
$$
and
$$
\widetilde\chi(x) = 1 \ \ \  \forall\ x\in[-3,3], \qquad \widetilde\chi(x) = 0 \ \  \ \forall\, x: \ |x| >4\, . 
$$ 
Take $0<\eps < \frac{1}{16}$ and define 
\begin{align} \label{s-kappa}
S_\kappa & = \widetilde\chi(\kappa^{-\eps} x) \left(-\pd_x^2 +\omega^2 x^2 -i\right)^{-1} + \left(1- \widetilde\chi(\kappa^{-\eps} x )\right)  \left(-\pd_x^2 +\Delta_\kappa(x) -i \right)^{-1} ,
\end{align}
where
$$
\Delta_\kappa(x) = \sqrt{\kappa}\, (v(0) -v ( \kappa^{-\frac 14}\, x) ) + \kappa^{2\eps}\, \chi(\kappa^{-\eps} x)\, .
$$
 Let
\begin{align} \label{eq-comm}
C_\kappa & := \big [h_\kappa\, ,\,  \widetilde\chi(\kappa^{-\eps} x) \big ] = -\kappa^{-2\eps}\, \widetilde\chi^{\, \prime\prime}(\kappa^{-\eps} x) - 2 \kappa^{-\eps} \, \widetilde\chi^{\, \prime}(\kappa^{-\eps} x)\, \pd_x  .
\end{align}
Using the fact that 
$$
\left(1- \widetilde\chi\left( \frac{x}{\kappa^{\eps}}\right)\right) \, \chi(\kappa^{-\eps} x) = 0 \qquad\forall\ x\in\R\, ,
$$
we obtain 
\begin{align}
(h_\kappa +\sqrt{\kappa}\, v(0)-i) S_\kappa  & = \widetilde\chi\left( \frac{x}{\kappa^{\eps}}\right)(h_\kappa +\sqrt{\kappa}\, v(0)-i) (-\pd_x^2 +\omega^2 x^2 -i)^{-1} \nonumber    \\
&  \quad+ \left(1- \widetilde\chi\left( \frac{x}{\kappa^{\eps}}\right)\right)(h_\kappa +\sqrt{\kappa}\, v(0)-i)  \left(-\pd_x^2 +\Delta_\kappa(x) -i \right)^{-1} \nonumber \\
&  \quad+ C_\kappa \left( (-\pd_x^2 +\omega^2 x^2 -i)^{-1} + (-\pd_x^2 +\Delta_\kappa(x) -i)^{-1}\right) \nonumber \\
& = \widetilde\chi\left( \frac{x}{\kappa^{\eps}}\right)(h_\kappa +\sqrt{\kappa}\, v(0)-i) (-\pd_x^2 +\omega^2 x^2 -i)^{-1} 
 + \left(1- \widetilde\chi\left( \frac{x}{\kappa^{\eps}}\right)\right) \nonumber \\
& \quad + C_\kappa \left( (-\pd_x^2 +\omega^2 x^2 -i)^{-1} + (-\pd_x^2 +\Delta_\kappa(x) -i)^{-1}\right) \label{comm1} \, .
\end{align} 
 From Taylor's formula with remainder applied to $v$, given any $x\in\R$ one can find $t_x\in \R $ such that 
\begin{equation} \label{taylor}
v(0) -v( \kappa^{-\frac 14}\, x) = \omega^2 x^2 \kappa^{-\frac 12}  + \frac{v'''(t_x)}{6}\ x^3 \kappa^{-\frac 34} \, . 
\end{equation}
This together with the definition of $\chi$ implies that there exists $c_1>0$, independent of $\kappa$, such that
\begin{equation} \label{d-lowerb}
\Delta_\kappa(x) \, \geq\, c_1\, \kappa^{2\eps}
\end{equation}
holds for all $x\in\R$ and all $\kappa\geq 1$, see Lemma \ref{lem-app}. Hence in view of \eqref{eq-comm} 
$$
 \|\, C_\kappa \left( (-\pd_x^2 +\omega^2 x^2 -i)^{-1} + (-\pd_x^2 +\Delta_\kappa(x) -i)^{-1}\right) \|_{L^2(\R)\to L^2(\R)}  \ \leq \ c_2 \, \kappa^{-\eps}
$$
holds for all $\kappa\geq 1$ and some $c_2>0$ independent of $\kappa$. 
To control the remaining term in \eqref{comm1} we use again the expansion \eqref{taylor} and note that 
$$
h_\kappa +\sqrt{\kappa}\, v(0) = -\pd_x^2 +\omega^2 x^2 +\frac{v'''(t_x)}{6}\, x^3\, \kappa^{-\frac 14}\, ,
$$
which implies 
\begin{align*}
 \widetilde\chi\left( \frac{x}{\kappa^{\eps}}\right)(h_\kappa +\sqrt{\kappa}\, v(0)-i) (-\pd_x^2 +\omega^2 x^2 -i)^{-1}  & =  \widetilde\chi\left( \frac{x}{\kappa^{\eps}}\right) + T_{\kappa,1}
 \end{align*}
where
$$
\| T_{\kappa,1}\|_{L^2(\R)\to L^2(\R)}   \ \leq\  c_3\, \kappa^{-\eps}, \qquad \kappa\geq 1, 
$$
and $c_3>0$ is a constant independent of $\kappa$.
Putting the above estimates together we conclude that 
\begin{equation}
(h_\kappa +\sqrt{\kappa}\, v(0)-i)\, S_\kappa = \id + T_\kappa, 
\end{equation} 
with $T_\kappa$ satisfying the estimate
\begin{equation} \label{t-upperb}
 \| T_\kappa\|_{L^2(\R)\to L^2(\R)}   \ \leq\  (c_2+c_3) \, \kappa^{-\eps} \qquad \forall\ \kappa\geq 1\, . 
\end{equation} 
Hence for $\kappa$ large enough the operator $\id + T_\kappa$ is invertible and the Neumann series for $(\id + T_\kappa)^{-1}$ shows that
\begin{equation} \label{hs}
\lim_{\kappa\to\infty}\,  \|(h_\kappa +\sqrt{\kappa}\, v(0)-i)^{-1} - S_\kappa \|_{L^2(\R)\to L^2(\R)} =0 .
\end{equation} 
On the other hand, since the multiplication operator $\widetilde\chi(\kappa^{-\eps} x)-1$ converges strongly to zero in $L^2(\R)$ as $\kappa\to\infty$ and $   (-\pd_x^2  +\omega^2 x^2 -i)^{-1}$ is compact, it follows that 
$$
\big\|\, \widetilde\chi(\kappa^{-\eps} x)   (-\pd_x^2 +\omega^2 x^2 -i)^{-1}  -(-\pd_x^2 +\omega^2 x^2 -i)^{-1} \big\|_{L^2(\R)\to L^2(\R)} \ \to \ 0 
$$
as $\kappa\to\infty$. Hence in view of \eqref{s-kappa}  and \eqref{d-lowerb} 
\begin{equation} \label{S-osc}
\lim_{\kappa\to\infty}\,  \|\, S_\kappa - (-\pd_x^2 +\omega^2 x^2 -i)^{-1} \|_{L^2(\R)\to L^2(\R)}= 0. 
\end{equation} 
This in combination with \eqref{hs} implies that $h_\kappa +\sqrt{\kappa}\, v(0)$ converges in the norm resolvent sense to 
$-\pd_x^2 +\omega^2 x^2$ and the claim follows.
\end{proof}

\begin{lem} \label{lem-2} 
Let $\psi_\kappa$ be the positive eigenfunction of $h_\kappa$ associated to the eigenvalue $E_1(\kappa)$ and normalized such that $\|\psi_\kappa\|=1$ for all $\kappa>0$. Then there exist $\alpha>0$ and $\kappa_2\geq1$ such that 
\begin{equation} \label{exp}
\sup_{\kappa\geq \kappa_2} \int_\R e^{2\alpha\sqrt{1+x^2}}\ |\psi_\kappa(x)|^2\, dx < \infty.
\end{equation}
\end{lem}

\begin{proof}
For any $z\in\C$ and $f\in C_0^\infty(\R)$ we have
$$
e^{\alpha\sqrt{1+(\cdot)^2}}\, (h_\kappa-z) \, e^{-\alpha\sqrt{1+(\cdot)^2}} f\, =( h_\kappa -z + W_\alpha) f, 
$$
where
\begin{equation}
W_\alpha := \frac{2\alpha x}{\sqrt{1+x^2}}\ \pd_x + \frac{\alpha}{(1+x^2)^{3/2}} - \frac{\alpha^2 x^2}{1+x^2}\, .
\end{equation} 
This shows that for any $u\in D(h_\kappa)= H^2(\R)$ and every $\alpha\in (0,1)$
\begin{equation}\label{W-upperb}
\|W_\alpha \, u\|^2 \ \leq \ 4\alpha^2\, \|\pd_x u\|^2 + 4\alpha^2\, \|u\|^2\, .
\end{equation}
In particular, this shows that $h_\kappa + W_\alpha$ is closed on the domain of $h_\kappa$. Next we define the curve 
\begin{equation}
\Gamma := \{ z\in\C:\, |z+ \sqrt{\kappa}\, v(0)-\omega| = \omega  \}
\end{equation}
By Lemma \ref{lem-1} there exists $\kappa_\omega$ such that 
\begin{equation} 
\sup_{z\in\Gamma}\,  \sup_{\kappa\geq \kappa_\omega}\, \|(h_\kappa -z )^{-1} \|\  < \ \infty. 
\end{equation}
On the other hand, since $h_\kappa + \sqrt{\kappa}\, v(0) \geq -\pd_x^2\, $ in the sense of quadratic forms, for any $z\in\Gamma$ we have 
\begin{align*}
\big | \lef (h_\kappa-z) u, u \rig_{L^2(\R)} \big | &\ \geq \ {\rm Re\, } \lef (h_\kappa-z) u, u \rig_{L^2(\R)} \ \geq\ \|\pd_x u\|^2 -2\omega \|u\|^2.
\end{align*}
This in combination with \eqref{W-upperb} implies that 
\begin{align*}
\|W_\alpha\, u\|^2  & \ \leq \ 4\alpha^2 (1+2\omega)\, \|u\|^2 + 4\alpha^2 \big | \lef (h_\kappa-z) u, u \rig_{L^2(\R)} \big | \\
& \ \leq \ 4\alpha^2 (1+2\omega)\, \|u\|^2 + 4\alpha^2\, \|(h_\kappa-z) u\|\, \|u\| \\
& \ \leq \ 8\alpha^2 (1+\omega)\, \|u\|^2 + 2\alpha^2\, \|(h_\kappa-z) u\|^2\, 
\end{align*}
for all $z\in\Gamma$. Hence by \cite[Thm.~IV.1.16]{ka} there exists $\alpha \in(0,1)$ small enough such that the operator 
$ h_\kappa + W_\alpha-z 
$
is invertible for all $z\in\Gamma$ and all $\kappa\geq\kappa_\omega$, with a bounded inverse. Then one can prove the identity:
$$
\,\, (h_\kappa-z)^{-1} \, e^{-\alpha\sqrt{1+(\cdot)^2}}=  e^{-\alpha\sqrt{1+(\cdot)^2}}(h_\kappa + W_\alpha-z)^{-1} \, ,
$$
which shows that 
\begin{equation}\label{sup-sup}
\sup_{z\in\Gamma}\,  \sup_{\kappa\geq \kappa_\omega}\, \big\| e^{\alpha\sqrt{1+(\cdot)^2}}\, (h_\kappa -z )^{-1}\, e^{-\alpha\sqrt{1+(\cdot)^2}}\ \big\|\ < \infty. 
\end{equation} 
Now denote by 
\begin{equation} \label{pk}
P_\kappa = \psi_\kappa \lef\, \cdot\, , \psi_\kappa \rig_{L^2(\R)}
\end{equation}
the projection on the eigenspace of $h_\kappa$ associated to $E_1(\kappa)$. Then by Lemma \ref{lem-1} and equation \eqref{sup-sup} 
\begin{align} \label{kappa-omega}
\sup_{\kappa\geq \kappa_\omega}\,  \big\| e^{\alpha\sqrt{1+(\cdot)^2}}\, P_\kappa\, e^{-\alpha\sqrt{1+(\cdot)^2}}\ \big\| & \ \leq \ \sup_{\kappa\geq \kappa_\omega}\, \frac{1}{2\pi} \oint_\Gamma  \big\| e^{\alpha\sqrt{1+(\cdot)^2}}\, (h_\kappa -z )^{-1}\, e^{-\alpha\sqrt{1+(\cdot)^2}}\ \big\|\, dz \ < \infty. 
\end{align}
To continue we denote by $\phi_1$ the normalized ground state of the harmonic oscillator  $-\pd_x^2 + \omega^2 x^2$. Let $\chi_R$ be the characteristic function of the interval $[-R,R]$. Since $\|\phi_1\|_{L^2(\R)}=1$, there exists $R$ large enough such that $\lef \phi_1,  \chi_R\, \phi_1\rig_{L^2(\R)} \geq 3/4$. On the other hand, from the proof of Lemma \ref{lem-1}, see equations \eqref{hs} and \eqref{S-osc}, it follows that $\psi_\kappa$ converges strongly to $\phi_1$ in $L^2(\R)$ as $\kappa\to\infty$. Therefore 
$$
\lef \psi_\kappa,  \chi_R\, \psi_\kappa\rig_{L^2(\R)}\  \geq \frac 12
$$
holds true for all $\kappa \geq \kappa_R$, where $\kappa_R$ depends only on the (fixed) value of $R$. Writing 
$$
\psi_\kappa = \frac{ P_\kappa (\psi_\kappa \, \chi_R)}{\lef \psi_\kappa,  \chi_R\, \psi_\kappa\rig_{L^2(\R)}}
$$
we thus conclude with the estimate
\begin{align*}
 \|\, e^{\alpha\sqrt{1+x^2}}\, \psi_\kappa\|_{L^2(\R)} & \leq 2\,   \big\| \, e^{\alpha\sqrt{1+(\cdot)^2}}\, P_\kappa\, e^{-\alpha\sqrt{1+(\cdot)^2}}\, \big\|\,  \big\| e^{\alpha\sqrt{1+(\cdot)^2}} \chi_R\, \psi_\kappa\, \big\|_{L^2(\R)} \\
  & \leq 2  \, e^{\alpha\sqrt{1+R^2}} \, \big\|\,  e^{\alpha\sqrt{1+(\cdot)^2}}\, P_\kappa\, e^{-\alpha\sqrt{1+(\cdot)^2}}\, \big\| \, ,
\end{align*}
which holds for all  $\kappa \geq \kappa_R$. Hence in view of \eqref{kappa-omega} 
$$
\sup_{\kappa\geq \kappa_2} \  \|\, e^{\alpha\sqrt{1+x^2}}\, \psi_\kappa\|_{L^2(\R)}  \, < \infty, 
$$ 
where $ \kappa_2 = \max\{\kappa_R, \kappa_\omega\}$. 
\end{proof}


\vspace{0.2cm}

\subsection{ Proof of Theorem \ref{thm-strong}}
We will prove the absence of discrete eigenvalues of $\h_{\kappa, \lambda}$ when $\kappa$ is larger than some critical value depending on $\lambda>1$. By Proposition \ref{prop-es} we have
\begin{equation}
\sigma_{\rm es} (\h_{\kappa, \lambda}) = [E_1(\kappa), \infty) \qquad \forall \ \kappa \geq 1, \ \forall\ \lambda>0. 
\end{equation}
Hence from \eqref{ham-scaled} and  perturbation theory the claim follows if we can show that
\begin{equation} \label{enough} 
\h_{\kappa, \lambda} -z \quad \text{is  invertible} \qquad \forall \ z \in \Big [ E_1(\kappa) -\frac{v(0)}{\sqrt{\kappa}}\, , \, E_1(\kappa) \Big)\, .
\end{equation}
Define the projection $\Pi_\kappa$  on $L^2(\R^2)$ by 
$$
\Pi_\kappa = P_\kappa \otimes \id_y, 
$$
where $P_\kappa$ is given by \eqref{pk} and $\id_y$ denotes the identity operator in $L^2(\R)$. Let $\Pi_\kappa^\perp = \id -\Pi_\kappa$.  
Then, according to the Feshbach-Schur formula \cite{schur}, \eqref{enough} is equivalent to proving that 
 for all 
\begin{equation} \label{z-int}
 z \in \Big [ E_1(\kappa) -\frac{v(0)}{\sqrt{\kappa}}\, , \, E_1(\kappa)   \Big),  
\end{equation}
the operator
\begin{equation} \label{ff-1}
\Pi_\kappa^\perp (\h_{\kappa, \lambda} -z) \Pi_\kappa^\perp  \quad  \text{is  invertible in} \ {\rm Ran}(\Pi_\kappa^\perp),
\end{equation} 
and at the same time, the operator 
\begin{equation} \label{ff-2}
\Pi_\kappa (\h_{\kappa, \lambda} -z) \Pi_\kappa - \Pi_\kappa \h_{\kappa, \lambda} \Pi_\kappa^\perp \big( \Pi_\kappa^\perp \, (\h_{\kappa, \lambda} -z)  \, \Pi_\kappa^\perp\big)^{-1} \, \Pi_\kappa^\perp \h_{\kappa, \lambda} \Pi_\kappa
  \ \   \text{is  invertible in} \ {\rm Ran}(\Pi_\kappa).
\end{equation} 

As for \eqref{ff-1} we note that by \eqref{ham-scaled}  
\begin{align*}
\Pi_\kappa^\perp \, (\h_{\kappa, \lambda} - z) \, \Pi_\kappa^\perp\ \geq \ (E_2(\kappa)-E_1(\kappa) -v(0) \, \kappa^{-\frac 12})\,  \Pi_\kappa^\perp \, , 
\end{align*}
Hence in view of Lemma \ref{lem-1} there exists $C_1>0$, independent of $\kappa$ and $z$, such that 
\begin{equation} \label{ff-1-ok}
\big\| \big( \Pi_\kappa^\perp\, (\h_{\kappa, \lambda} -z) \, \Pi_\kappa^\perp\big)^{-1} \big\|_{{\rm Ran}(\Pi_\kappa^\perp) \to {\rm Ran}(\Pi_\kappa^\perp)} \ \leq \ C_1
\end{equation}
holds for all $\kappa$ large enough and all $z$ satisfying \eqref{z-int}.

\medskip

\noindent In order to prove \eqref{ff-2} we note that 
$$
v(\kappa^{-\frac 14}\, (x-y)) = v(\kappa^{-\frac 14}\, y) + \kappa^{-\frac 14} \int_0^x v'(\kappa^{-\frac 14}\, (t-y))\, dt,
$$
which together with \eqref{ham-scaled} yields 
\begin{equation} 
\Pi_\kappa\, \h_{\kappa, \lambda}\, \Pi_\kappa\,  = -\pd_y^2+E_1(\kappa) + \frac{\lambda-1}{\sqrt{\kappa}}\, v(\kappa^{-\frac 14}\, y) -\kappa^{-\frac 34} \int_\R \psi^2(x) \int_0^x v'(\kappa^{-\frac 14}\, (t-y))\, dt dx .
\end{equation} 
Denote by
\begin{equation} \label{kz}
K_z(x,y,x',y') = \big( \Pi_\kappa^\perp\, (\h_{\kappa, \lambda} -z) \, \Pi_\kappa^\perp\big)^{-1/2}(x,y,x',y')
\end{equation}
the Schwartz integral kernel of $\big( \Pi_\kappa^\perp\, (\h_{\kappa, \lambda} -z) \, \Pi_\kappa^\perp\big)^{-1/2}$. 
To treat the second term in \eqref{ff-2} we introduce the bounded operator $A: L^2(\R) \to L^2(\R^2)$ corresponding to:
\begin{equation} \label{A-def}
(A f)(x,y) = \int_{\R^2} K_z(x,y,x',y')\ v(\kappa^{-\frac 14}\, (x'-y'))\, \psi_\kappa(x') f(y') \, dx' dy', 
\end{equation} 
where the above formal expression has to be first understood as a map from the Schwartz space $S(\R)$ to the dual of $S(\R^2)$, which can afterwards be extended to a bounded operator between $L^2(\R)$ to $L^2(\R^2)$. An important role here is played by the estimate:
$$
\int_{\R^2} | v(\kappa^{-\frac 14}\, (x'-y'))\, \psi_\kappa(x') f(y')| ^2\, dx' dy' \ \leq \ v(0)^2 \int_\R f^2(y')\, dy' . 
$$
Another important observation is that $A$ is also bounded from $L^\infty(\R)$ to $L^2(\R^2)$ due to the inequality:
$$
\int_{\R^2} | v(\kappa^{-\frac 14}\, (x'-y'))\, \psi_\kappa(x') f(y')| ^2\, dx' dy' \ \leq \kappa^{1/4} \|v\|_2^2\|f\|_\infty^2. 
$$

Hence $A$ is bounded uniformly with respect to $z\in [E_1(\kappa) -v(0)/\sqrt{\kappa}\, , \, E_1(\kappa))$ in view of \eqref{ff-1-ok}, and we can write:
$$
\Pi_\kappa\, \h_{\kappa, \lambda} \, \Pi_\kappa^\perp \big( \Pi_\kappa^\perp \, (\h_{\kappa, \lambda} -z)  \, \Pi_\kappa^\perp\big)^{-1} \, \Pi_\kappa^\perp \, \h_{\kappa, \lambda} \, \Pi_\kappa = \frac 1\kappa A^*A\, .
$$
This implies: 
\begin{align*}
 \Pi_\kappa\, (\h_{\kappa, \lambda}-z) \, \Pi_\kappa & - \Pi_\kappa\, \h_{\kappa, \lambda} \, \Pi_\kappa^\perp \big( \Pi_\kappa^\perp \, (\h_{\kappa, \lambda} -z)  \, \Pi_\kappa^\perp\big)^{-1} \, \Pi_\kappa^\perp \, \h_{\kappa, \lambda} \, \Pi_\kappa= \\
 & =  \Pi_\kappa\, (B_{\kappa, \lambda} +E_1(\kappa)-z)\,  \Pi_\kappa\, ,  
\end{align*}
where
\begin{align} \label{hc1}
B_{\kappa, \lambda} &:=
-\pd_y^2+ \frac{\lambda-1}{\sqrt{\kappa}}\, v(\kappa^{-\frac 14}\, y) -\kappa^{-\frac 34} \int_\R \psi^2(x) \int_0^x v'(\kappa^{-\frac 14}\, (t-y))\, dt dx  - \frac 1\kappa A^*A\, .
\end{align}
To prove \eqref{enough} it thus suffices to show that 
\begin{equation} \label{enough-2}
B_{\kappa, \lambda} + \xi \quad \text{is  invertible in} \  \ L^2(\R) \ \  \text{for all} \  \ \xi  \in \big (0, \kappa^{-\frac12}\, v(0) \big] \, .
\end{equation} 
To this end we write 
\begin{align} \label{B-1}
B_{\kappa, \lambda} +\xi = -\pd_y^2 + \xi +a_1\, a_2 +b_1\, b_2 + d_1\, d_2 \, ,
\end{align} 
where 
$$
d_1 =-\frac 1\kappa A^* \, , \qquad d_2 = A, 
$$
and $a_1, a_2, b_1$ and $b_2$ are multiplication operators in $L^2(\R)$ given by 
\begin{align}
a_1(y) &  =  \frac{\lambda-1}{\sqrt{\kappa}}\, \sqrt{v(\kappa^{-\frac 14}\, y)},   \quad a_2 = \sqrt{v(\kappa^{-\frac 14}\, y)}, \quad b_1(y)  = -\kappa^{-\frac 34}  \Big |\, \psi^2(x) \int_0^x v'(\kappa^{-\frac 14}\, (t-y))\, dt dx\Big |^{\frac 12} \nonumber  \\
 b_2(y)  &  =   \Big |\, \psi^2(x) \int_0^x v'(\kappa^{-\frac 14}\, (t-y))\, dt dx\Big |^{\frac 12} \ {\rm sign} \Big(\psi^2(x)\! \int_0^x v'(\kappa^{-\frac 14}\, (t-y))\, dt dx\Big).   \label{ab}
\end{align}
Let
\begin{equation} \label{Q}
Q_{\kappa, \lambda}(\xi) = \id
+
\left( \begin{array}{ccc}
a_2(-\pd_y^2+\xi)^{-1} a_1 & a_2(-\pd_y^2+\xi)^{-1} b_1 & a_2(-\pd_y^2+\xi)^{-1} d_1 \\
b_2(-\pd_y^2+\xi)^{-1} a_1 & b_2(-\pd_y^2+\xi)^{-1} b_1 & b_2(-\pd_y^2+\xi)^{-1} d_1 \\
d_2(-\pd_y^2+\xi)^{-1} a_1 & d_2(-\pd_y^2+\xi)^{-1} b_1 & d_2(-\pd_y^2+\xi)^{-1} d_1 
\end{array} \right)
\end{equation} 
be a matrix-valued operator in $L^2(\R) \oplus L^2(\R) \oplus L^2(\R^2)$. By \eqref{B-1} and the resolvent equation it follows that if 
$Q_{\kappa, \lambda}(\xi)$ is invertible for all $ \xi  \in \big (0, \kappa^{-\frac12}\, v(0) \big]$, then \eqref{enough-2} holds true and 
\begin{align*}
(B_{\kappa, \lambda} +\xi )^{-1} &= (-\pd_y^2+\xi)^{-1}-(-\pd_y^2+\xi)^{-1}\ (a_1, b_1, d_1)\, \big(Q_{\kappa, \lambda}(\xi)\big)^{-1}  (a_2, b_2, d_2)^{\rm T}\, (-\pd_y^2+\xi)^{-1}\, .
\end{align*}
Next we note that $(-\pd_y^2+\xi)^{-1}$ is an integral operator in $L^2(\R)$ with the kernel 
\begin{equation} \label{m}
(-\pd_y^2+\xi)^{-1} (y, y') = \frac{e^{-\sqrt{\xi}\ |y-y'|}}{2 \sqrt{\xi}} =: \frac{1}{2\sqrt{\xi}} + m(y,y'). 
\end{equation} 
Let $M$ be the integral operator in $L^2(\R)$ with the kernel $m(y,y')$ defined above and let 
$$
|\,  \Phi \rig := \Big (a_2(y) , b_2 (y) , \int_\R d_2(x,y;y')\, dy' \Big)^{\rm T} , \quad \lef\Psi\, | := \Big (a_1(y) , b_1 (y) , \int_\R d_1(x;y,y')\, dx \Big), 
$$
where $d_1(x;y, y')$ and $d_2(x,y;y')$ are the integral kernels of $A^*$ and $A$ respectively. The integrals involving the integral kernels make sense because $A$ is bounded from $L^\infty(\R)$ to $L^2(\R^2)$. Equations \eqref{Q} and \eqref{m} then imply that
\begin{equation} \label{Q-2}
Q_{\kappa, \lambda}(\xi) = \id +  R_{\kappa,\lambda}+ \frac{1}{2\sqrt{\xi}} \   |\, \Phi \rig \lef\Psi\, | 
\end{equation} 
holds true with 
\begin{equation} \label{eq-R}
R_{\kappa,\lambda} = 
\left( \begin{array}{ccc}
a_2M a_1 & a_2M b_1 & a_2M d_1 \\
b_2M a_1 & b_2M b_1 & b_2M d_1 \\
d_2M a_1 & d_2M b_1 & d_2M d_1 
\end{array} \right).
\end{equation}
To prove the invertibility of $Q_{\kappa, \lambda}(\xi)$ we first show that $\id+R_{\kappa,\lambda}$ is invertible for $\lambda-1$ sufficiently small, but positive, and $\kappa$ sufficiently large, uniformly with respect to $\xi>0$. To do so we estimate  
the operator norm of all the entries of $R_{\kappa,\lambda}$ keeping in mind that the integral kernel of $M$ satisfies
\begin{equation} \label{m-upperb}
|\, m(y,y')\, | =  \frac{1-e^{-\sqrt{\xi}\ |y-y'|}}{2 \sqrt{\xi}}   \ \leq \ |y-y'| \, .
\end{equation}
To simplify the notation in the sequel we introduce the following shorthands; 
\begin{align} \label{notation}
m_{j,k} & = \int_\R x^j\, v^k(x)\, dx, \qquad   m'_{j,k}  = \int_\R |x|^j\, |v'(x)|^k\, dx, \qquad \mu_j = \sup_{\kappa\geq \kappa_2} \int_\R |x|^j\, \psi_\kappa^2(x)\, dx, 
\end{align}
where $j,k\in\N$ and $\kappa_2$ is given by Lemma \ref{lem-2}. We start with the first column of $R_{\kappa,\lambda}$. Using \eqref{m-upperb} we get
\begin{align*}
\|a_2\, M\, a_1\|_{L^2(\R)\to L^2(\R)}^2 & \ \leq \ \|a_2\, M\, a_1\|_{\rm HS}^2 \ \leq \ \frac{(\lambda-1)^2}{\kappa} \int_{\R^2}\, v(\kappa^{-\frac 14}\, y) \, |y-y'|^2\, v(\kappa^{-\frac 14}\, y')\, dy dy' \\
& \leq \ \frac{(\lambda-1)^2}{\kappa} \int_{\R^2}\, v(\kappa^{-\frac 14}\, y)\,  (2y^2+2y'^2)\, \, v(\kappa^{-\frac 14}\, y')\, dy dy' \\
& = 4(\lambda-1)^2 \, m_{0,1}\, m_{2,1} \, .
\end{align*}
In the same way, using Lemma \ref{lem-app-2}, it follows that 
\begin{align}
\|b_2\, M\, a_1\|_{L^2(\R)\to L^2(\R)}^2 &\  \leq \ \frac{(\lambda-1)^2}{\kappa} \int_{\R^2}\, b_2^2(y)\,  (2y^2+2y'^2)\, \, v(\kappa^{-\frac 14}\, y')\, dy dy' \nonumber \\
& = \ \frac{2 (\lambda-1)^2}{\kappa^{3/4}} \, \left(m_{0,1} \int_{\R}\, b_2^2(y)\, y^2 dy +  m_{2,1} \int_{\R}\, b_2^2(y)\,  dy\right)\nonumber \\
& \leq \ 2 (\lambda-1)^2\, m_{0,1}\, m'_{2,1}\, \mu_1 + \frac{2 (\lambda-1)^2}{\sqrt{\kappa}}\ (m_{0,1}\, m'_{0,1}\, \mu_3+ m_{2,1}\, m'_{2,1}\, \mu_1) \nonumber \\
& \leq \ C_{ba}\, (\lambda-1)^2 \qquad \forall\ \kappa\geq \kappa_2, \label{ba-upperb}
\end{align}
where $C_{ba}$ is a constant independent of $\lambda$ and $\kappa$. Similarly,  
\begin{align*}
\|b_2\, M\, b_1\|_{L^2(\R)\to L^2(\R)}^2 &\  \leq \ \kappa^{-\frac 32} \int_{\R^2}\, b_2^2(y)\,  (2y^2+2y'^2)\, \, b_2^2(y')\, dy dy' \\
& \leq \ C_{bb}\, \kappa^{-\frac 12} \qquad \forall\ \kappa\geq \kappa_2,
\end{align*}
where we used Lemma \ref{lem-app-2} and \eqref{ab}. 

\medskip

\noindent To estimate $d_2M a_1$ we first observe that for any $f\in L^2(\R)$ 
\begin{equation} \label{d2-a1}
(d_2M a_1 f)(x,y) = \frac{\lambda-1}{\sqrt{\kappa}}\, \int_{\R^2} \, K_z(x,y,x',y')\, u(x',y')\, dx' dy',
\end{equation}
where $K_z(x,y,x',y')$ is given by \eqref{kz} and 
$$
u(x',y') = \psi_\kappa(x')  v(\kappa^{-\frac 14}\, (x'-y')) \int_\R m(y', y'') \sqrt{ v(\kappa^{-\frac 14}\, y'')}\ f(y'')\, dy''
$$
Hence by the H\"older inequality and Lemma \ref{lem-app-2} 
\begin{align*}
\|u\|_{L^2(\R^2)}^2 & \leq \ \|f\|_{L^2(\R)}^2\ \int_{\R^2} \psi^2_\kappa(x') \,  v^2(\kappa^{-\frac 14}\, (x'-y')) \left(\int_\R (2y'^2+2y''^2) v(\kappa^{-\frac 14}\, y'') \, dy'' \right)\, dx' dy' \\
& = 2 \|f\|_{L^2(\R)}^2\  \int_{\R^2} \psi^2_\kappa(x') \,  v^2(\kappa^{-\frac 14}\, (x'-y')) (y'^2  \kappa^{\frac 14}\, m_{0,1}+ \kappa^{\frac 34}\, m_{2,1} \ )\, dx' dy' \\
& \leq \ 2 \|f\|_{L^2(\R)}^2 \Big [ 2 \kappa^{\frac 14}\, m_{0,1} (\kappa^{\frac 34}\, m_{2,2} +  \kappa^{\frac 14}\, m_{0,2}\, \mu_2) + \kappa\, m_{2,1}\, m_{0,2} )\Big]
\end{align*}
Therefore, in view of \eqref{ff-1-ok} and \eqref{d2-a1} there exists a constant $C_{da}$, independent of $\kappa$, such that 
\begin{equation} \label{da-upperb}
\|d_2\, M\, a_1\|_{L^2(\R)\to L^2(\R^2)}\  \leq\ C_{da}\, (\lambda-1)\qquad \forall\ \kappa\geq \kappa_2.
\end{equation}
In the same way it follows that
\begin{equation} \label{db-upperb}
\|d_2\, M\, b_1\|_{L^2(\R)\to L^2(\R^2)}\  \leq\ C_{db}\, \kappa^{-\frac 14} \qquad \forall\ \kappa\geq \kappa_2.
\end{equation}
As for the operator  $d_2M d_1: L^2(\R^2)\to L^2(\R^2)$ we note that for any $f\in L^2(\R)$ it holds 
\begin{equation} \label{d2d1-def}
(d_2 \, M\, d_1 f)(x,y) = -\frac 1\kappa\, \big( \Pi_\kappa^\perp\, (\h_{\kappa, \lambda} -z) \, \Pi_\kappa^\perp\big)^{-1/2} \, \psi_\kappa(x) \, v(\kappa^{-\frac 14}\, (x-y))\, u(y),
\end{equation}
where
$$
u(y) = \int_{\R^4} \psi_\kappa(y') \, m(y,y')\, v(\kappa^{-\frac 14}\, (y'-t))\, K_z(y',t,x',s)\, f(x',s)\, dx' ds dt dy'
$$
From \eqref{d2d1-def} and H\"older inequality we then obtain 
\begin{align*}
 & \|d_2\, M\, d_1\|^2_{L^2(\R^2)\to L^2(\R^2)}  \leq \frac{1}{\kappa^2}\ \big\|\big( \Pi_\kappa^\perp\, (\h_{\kappa, \lambda} -z) \, \Pi_\kappa^\perp\big)^{-1/2}\big\|^2   \\
&\qquad \qquad\quad   \int_{\R^4} \psi^2_\kappa(x) \, \psi^2_\kappa(t) \, (s-t)^2\,  v^2(\kappa^{-\frac 14}\, (x-s)) \, v^2(\kappa^{-\frac 14}\, (t-y))\, dx dy ds dt \\
& \qquad \quad =  C\,  m_{0,2}\, \kappa^{-\frac 74} \, \int_{\R^3} \psi^2_\kappa(x) \, \psi^2_\kappa(t) \, (s-x+x-t)^2\,  v^2(\kappa^{-\frac 14}\, (x-s)) \, dx ds dt \\
&\qquad \quad   \leq \ 2\, C\,  m_{0,2}\, \kappa^{-\frac 74}  \, \int_{\R^3} \psi^2_\kappa(x) \, \psi^2_\kappa(t) \, \big((s-x)^2+(x-t)^2\big)\,  v^2(\kappa^{-\frac 14}\, (x-s)) \, dx ds dt \\
& \qquad \quad  = 2\, C\,  m_{0,2}\, \kappa^{-\frac 74}  \, \left(\kappa^{\frac 34} \, m_{2,2} + \int_{\R^3} \psi^2_\kappa(x) \, \psi^2_\kappa(t) \, (x^2+t^2)\,  v^2(\kappa^{-\frac 14}\, (x-s)) \, dx ds dt \right) \\
&  \qquad \quad = 2\, C\,  m_{0,2}\, \kappa^{-\frac 74} \, \left(\kappa^{\frac 34} \, m_{2,2} +2 \kappa^{\frac 14}\, m_{0,2}\, \mu_2\right)\, ,
\end{align*}
where we have used Lemma \ref{lem-app-2} again and the fact that $\psi_\kappa(\cdot)$ is even. Hence there exists a constant $C_d$ such that 
\begin{equation} 
 \|d_2\, M\, d_1\|_{L^2(\R^2)\to L^2(\R^2)}\ \leq \ C_d\, \kappa^{-\frac 12} \qquad \forall \ \kappa\geq \kappa_2\, .
\end{equation}
Concerning the remaining entries of $R_{\kappa,\lambda}$, we notice that by duality, \eqref{ab} and \eqref{da-upperb}
$$
\|a_2\, M\, d_1\|_{L^2(\R^2)\to L^2(\R)} = \frac{1}{(\lambda-1)\, \sqrt{\kappa}}\ 
\|\, d_2\, M\, a_1\|_{L^2(\R)\to L^2(\R^2)}\  \leq\ C_{da}\, \kappa^{-\frac 12} \qquad \forall\ \kappa\geq \kappa_2.
$$
Similarly it follows from \eqref{ba-upperb} and \eqref{db-upperb} that
$$
\|a_2\, M\, b_1\|_{L^2(\R)\to L^2(\R)} =  \frac{1}{(\lambda-1)\, \kappa^{1/4} }\  \|b_2\, M\, a_1\|_{L^2(\R)\to L^2(\R)} \ \leq \ C_{ba}\, \kappa^{-\frac 14} \qquad \forall\ \kappa\geq \kappa_2,
$$
and
$$
\|b_2\, M\, d_1\|_{L^2(\R^2)\to L^2(\R)} = \kappa^{-\frac 14}  \|d_2\, M\, b_1\|_{L^2(\R)\to L^2(\R^2)} \ \leq \ C_{bb}\, \kappa^{-\frac 12} \qquad \forall\ \kappa\geq \kappa_2.
$$
Putting together the above estimates we conclude that 
\begin{equation}\label{R-upperb0}  
 \|R_{\kappa,\lambda} \| \ \leq \ C_R\, (\, |\lambda-1| + \kappa^{-\frac 14} ) \qquad \forall  \lambda  >1, \ \ \forall \ \kappa\geq \kappa_2, 
\end{equation}
hold for some $C_R>0$, 
where the norm of $R_{\kappa,\lambda}$ is calculated on $L^2(\R) \oplus L^2(\R) \oplus L^2(\R^2)$. Hence there exists $1<\lambda_0<2$ (which has to be chosen close enough to $1$) and some $\kappa_0 \geq \kappa_2$ (independent of $\lambda_0$), such that 
\begin{equation}\label{R-upperb}  
 \|R_{\kappa,\lambda} \| \ \leq \ \frac 12 \qquad \forall \ \lambda \in (1, \lambda_0), \ \ \forall \ \kappa\geq \kappa_0
\end{equation} 
For these values of $\lambda$ and $\kappa$ the operator $\id + R_{\kappa,\lambda} $ is invertible, uniformly in $\xi$, and 
\eqref{Q-2} becomes:
\begin{align}\label{hc2}
Q_{\kappa,\lambda}(\xi)=\left (\id +  \frac{1}{2\sqrt{\xi}} \   |\, \Phi \rig \lef\Psi\, | (\id +  R_{\kappa,\lambda})^{-1}\right )(\id +  R_{\kappa,\lambda}),
\end{align}
hence we reduced the invertibility of $Q_{\kappa,\lambda}(\xi)$ to the one of 
$$\id +  \frac{1}{2\sqrt{\xi}} \   |\, \Phi \rig \lef\Psi\, | (\id +  R_{\kappa,\lambda})^{-1}.$$
After a second Feshbach-Schur reduction with respect to the projection on the vector $|\Phi\rangle$, we notice that this operator is invertible if and only if the function 
\begin{align}\label{hc3}
f_{\kappa,\lambda}(\xi):=1+\frac{1}{2\sqrt{\xi}}\lef \Psi\, , \, (\id +R_{\kappa,\lambda})^{-1}\Phi\rig, \quad  \xi  \in \big (0, \kappa^{-\frac12}\, v(0) \big] 
\end{align}
is never zero. The Neumann series for $ (\id +R_{\kappa,\lambda})^{-1}$ in combination with \eqref{R-upperb} gives
\begin{align}
\lef \Psi\, , \, (\id +R_{\kappa,\lambda})^{-1}\Phi\rig  &  =   \lef \Psi\, ,  \Phi\rig + \sum_{n=1}^\infty (-1)^n \lef \Psi\, , \, R^n_{\kappa,\lambda} \, \Phi\rig 
\nonumber \\ 
& \geq  \lef \Psi\, ,  \Phi\rig - \| \Phi\|\, \|\Psi\|\, \frac{\|R_{\kappa,\lambda} \| }{1-\|R_{\kappa,\lambda} \|} \nonumber\\
& \geq  \lef \Psi\, ,  \Phi\rig - 2\, \| \Phi\|\, \|\Psi\|\, \|R_{\kappa,\lambda} \|,  \label{neumann}
\end{align} 
where all the scalar products and norms are calculated on $L^2(\R) \oplus L^2(\R) \oplus L^2(\R^2)$. Equation \eqref{ff-1-ok} and a straightforward computation show that 
\begin{align}
 \lef \Psi\, ,  \Phi\rig & = (\lambda-1) \, \kappa^{-\frac14} \ m_{0,1} - \frac{1}{\kappa}\ {\rm tr} (A^*A)\label{phi-psi}\,  \geq \ (\lambda-1) \, \kappa^{-\frac14}\ m_{0,1} - C_1\, \kappa^{-\frac 34} \ m_{0,2},
\end{align} 
and that there exists a constant $C_0$ such that 
\begin{equation}
\| \Phi\|\, \|\Psi\|\, \leq \ C_0 \sqrt{\kappa^{-\frac 12}\ (\lambda-1)^2 + \kappa^{-1} }  \qquad \forall\ \kappa \geq \kappa_0.
\end{equation} 
Hence if we set 
\begin{equation}  \label{kc}
 \varkappa_c(\lambda) = \max \big\{ \kappa_0, (\lambda-1)^{-4} \, \big\} ,
\end{equation} 
then equations \eqref{R-upperb0}, \eqref{neumann} and \eqref{phi-psi} imply 
$$
\kappa^{1/4}\lef \Psi\, , \, (\id +R_{\kappa,\lambda})^{-1}\Phi\rig \geq  (\lambda-1) \, m_{0,1} +\mathcal{O}((\lambda-1)^2) \qquad \forall\ \lambda \in (1, \lambda_0),\  \ \forall\ \kappa \geq  \varkappa_c(\lambda).
$$
This shows that there exists $0<\lambda_c <\lambda_0<2$ such that
$$\lef \Psi\, , \, (\id +R_{\kappa,\lambda})^{-1}\Phi\rig\geq 0,
  \qquad \forall\ \lambda \in (1, \lambda_c), \  \ \ \ \forall\ \kappa \geq  \varkappa_c(\lambda).
$$ 
Thus $f_{\kappa,\lambda}(\xi)$ in \eqref{hc3} is never zero if $\lambda>1$ is close enough to $1$ and, at the same time, $\kappa$ is larger than some $\lambda$-dependent critical value. Since the number of discrete eigenvalues of $\h_{\kappa,\lambda}$ is non-increasing with respect to $\lambda$, we obtain the claim of the theorem for all $\lambda>1$. 
\qed

\vspace{0.2cm}

\subsection{Proof of Corollary \ref{cor-1}}

\noindent We know from Proposition \ref{prop-es}  that  $\HH_\kappa(v)$ and $\HH_{\kappa,3/2}(v)$ have the same essential spectrum if $\kappa>1$.  Due to Theorem \ref{thm-strong}, the discrete spectrum of $\HH_{\kappa,3/2}(v)$ is empty if $\kappa$ is larger than some critical value. Since $\HH_\kappa(v)\geq \HH_{\kappa,3/2}(v)$ for all $\kappa\geq 3/2$, the result follows from the min-max principle. 
\qed

\vspace{0.2cm}

\subsection{Proof of Proposition \ref{prop-well}}

We are interested in the case when $k>k_e$.
Let $\mathfrak{h}_{\kappa w} $ be the operator in $L^2(\R)$ given by 
\begin{equation}  \label{h-frak}
\mathfrak{h}_{\kappa w} = -\pd_x^2 -\kappa\, w(x)\,,
\end{equation} 
and let $e_w(\kappa)<0$ be its lowest eigenvalue. 
In view of Proposition \ref{prop-es} we have 
\begin{equation} \label{es-ex}
\sigma_{\rm es}(\HH_\kappa) = [e_w(\kappa), \infty), \qquad \forall\, \kappa \geq k_e. 
\end{equation}
We will construct a test function $u$ in the form 
\begin{equation} 
u(x,y) = \varphi_\kappa(x) \, f(y),
\end{equation}
where $\varphi_\kappa(x)$ is a normalized eigenfunction of the operator $\mathfrak{h}_{\kappa w} $ associated to its lowest eigenvalue $e_w(\kappa)$, and 
$$
f(y) = \begin{cases}
0
  & \text{if}\quad y\leq 1
  \\
y-1 & \text{if}\quad 1 < y < 2   \\

\exp(4-2y)  & \text{if} \quad  2 \leq y
  \,.
\end{cases}
$$
Integration by parts then shows that 
\begin{align}
Q_w[u] & :=  \lef u, \, \HH_\kappa(w)\,  u\rig_{L^2(\R^2)} -e_w(\kappa)\, \|u\|^2 \nonumber \\
&= \int_{\R^2}  \varphi^2_\kappa(x) |f'(y)|^2\, dx dy +\kappa \int_\R w(y) f^2(y)\, dy - \int_{\R^2} w(x-y) \, \varphi^2_\kappa(x) f^2(y)\, dx dy\nonumber
\\
& = 2 - \int_{\R^2} w(x-y) \, \varphi^2_\kappa(x) f^2(y)\, dx dy\,  \leq 2 - w_0 \int_1^{2} (y-1)^2 \int_{y-1}^{y+1}  \varphi^2_\kappa(x)\, dx dy \nonumber  \\
& \leq  2 -  w_0  \int_0^{1} \int_{1}^{x+1}(y-1)^2 \varphi^2_\kappa(x)\, dydx = 2 -  \frac{w_0}{3} \int_0^{1}   \varphi^2_\kappa(x)\, x^3\, dx .  \label{Q-fact-2}
\end{align}
On the other hand, an explicit calculation yields
\begin{equation}
 \varphi_\kappa(x) =  \begin{cases}
C_\kappa \cos (\beta_\kappa x) 
  & \text{for }\quad |x| \leq 1
  \,,
  \\
 D_\kappa \, e^{-|x|  \omega_\kappa\, } 
  & \text{for }  \, \quad  |x| > 1
  \, ,
\end{cases}
\end{equation} 
where 
\begin{equation} \label{beta-omega}
\beta_\kappa = \sqrt{\kappa\, w_0 +e_w(\kappa)} \ , \qquad \omega_\kappa = \sqrt{-e_w(\kappa)}\, , 
\end{equation} 
and $C_\kappa$ and $D_\kappa$ are constants satisfying 
\begin{align} 
C_\kappa   \cos (\beta_\kappa) & = D_\kappa \, e^{-\omega_\kappa} \label{cd-1} \\ 
\beta_\kappa\, C_\kappa   \sin (\beta_\kappa) & = \omega_\kappa\, D_\kappa \, e^{-\omega_\kappa}  \label{cd-2}.
 \end{align}
 The last two equations imply that $e_w(\kappa)$ is given by the smallest solution to the implicit equation 
\begin{equation} \label{eq-ompl}
\frac{\sqrt{-e_w(\kappa)}}{\beta_\kappa} = \tan (\beta_\kappa)\, ,
\end{equation} 
where
\begin{equation} \label{beta-eq}
0 < \  \beta_\kappa \ < \frac\pi 2.
\end{equation} 
From equations \eqref{Q-fact-2}, \eqref{beta-eq} and the elementary inequality $\sin^2 (x) \leq x^2$ we then obtain the upper bound
\begin{align}
Q_w[u] &\ \leq \ 2 -  \frac{w_0}{3}\, C_\kappa^2 \int_0^1   \left(1-\sin^2(\beta_\kappa x)\right) x^3\, dx \ \leq \ 
2 -  \frac{w_0}{3}\, C_\kappa^2 \int_0^{\frac 2\pi}   \left(1-\sin^2(\beta_\kappa x)\right) x^3\, dx  \nonumber\\
& = 2 -  \frac{w_0}{3}\, C_\kappa^2 \left( \int_0^{\frac 2\pi}  x^3\, dx - \beta_\kappa^2 \int_0^{\frac 2\pi}  x^5\, dx\right)\nonumber\ \\
& \leq 2 -\frac{4 w_0}{9 \pi^4}\, C_\kappa^2\, . \label{Q-upprb1}
\end{align} 
To prove that $Q_w[u]$ is negative for $w_0$ large enough we thus need a lower bound on $C_\kappa$ independent of $\kappa$.
The condition $\|\varphi_\kappa\|=1$ gives
$$
 C^2_\kappa \left(1+\beta^{-1}_\kappa \cos (\beta_\kappa) \sin (\beta_\kappa) \right) + \frac{D^2_\kappa}{\omega_\kappa} \, e^{-2\omega_\kappa} = 1 ,
$$
which in view of equations \eqref{beta-omega} and \eqref{cd-1}-\eqref{cd-2} implies 
\begin{equation} \label{ckappa-1}
C^2_\kappa \left[1+ \cos (\beta_\kappa) \sin (\beta_\kappa)\frac{\kappa\, w_0}{\beta_\kappa\, \omega^2_\kappa } \right] =1 .
\end{equation} 
Using $\sin(\beta_\kappa)/\beta_\kappa\leq 1$ in the above expression together with the identity $\omega_\kappa^2=-e_w(\kappa)$ we have:
\begin{equation} \label{ckappa-2}
1 \leq C^2_\kappa \left (1-\frac{\kappa\, w_0}{e_w(\kappa)} \right) .
\end{equation} 
 To continue we have to estimate $e_w(\kappa)$ from above. The  choice of the test function 
$$
\psi(x) =  \begin{cases}
1
  & \text{for }\quad |x| \leq 1
  \,,
  \\
  e^{(1-|x|) \kappa w_0}
  & \text{for }  \, \quad  |x| > 1
  \,.
\end{cases}
$$ 
gives 
$$
e_w(\kappa) \leq \frac{\langle \psi, \mathfrak{h}_{\kappa w} \psi\rangle}{\|\psi\|^2}= -\frac{\kappa w_0}{2+ \frac{1}{\kappa w_0}} \, < 0.
$$
Hence we get the upper bound
$$
- \frac{\kappa\, w_0}{e_w(\kappa)} \, \leq \, 2+ \frac{1}{\kappa w_0} \, .
$$
This in combination with \eqref{ckappa-2} leads to
$$C_\kappa^2\geq \left (3+\frac{1}{\kappa w_0}\right )^{-1},\quad \forall \kappa\geq k_e.$$
Because $k_e > 1/2$ (see Proposition \ref{prop-es}), the above estimate implies $$C_\kappa^2\geq \left (3+\frac{2}{ w_0}\right )^{-1},\quad \forall \kappa\geq k_e.$$
 Inserting this back into \eqref{Q-upprb1} leads to:
$$
Q_w[u] \, \leq   2 - \frac{4}{9 \pi^4}\, \frac{w_0^2}{3 w_0 + 2}\qquad \forall\, \kappa \geq  k_e. 
$$
Hence $Q_w[u] <0$ for $w_0$ large enough, uniformly in $\kappa \geq  k_e$, which ends the proof. 
\qed


\section{\bf Proofs in the two-dimensional case} \label{sect-4}

\noindent We introduce the scaling function  
\begin{equation}\label{eq-u2}
(\mathcal{U}_\kappa\,  f)(x,y) =  \kappa\,  f\big( \kappa^{\frac 12} x,  \kappa^{\frac 12} y),
\end{equation}
where $\mathcal{U}_\kappa$ maps $L^2(\R^4)$ unitarily onto itself, and define the operator 
\begin{equation} \label{A-2d} 
A_{\kappa} := \frac{1}{\kappa}\, \mathcal{U}^*_\kappa\, \hh_\kappa \, \mathcal{U}_\kappa = a_\kappa -\Delta_y  - V_{\rm ctr}( \kappa^{-\frac 12} y) + \frac 1\kappa V_{\rm ctr}(\kappa^{-\frac 12}|x-y|),
\end{equation} 
where 
\begin{equation} \label{ak}
a_k = -\Delta +V_{\rm ctr}( \kappa^{-\frac 12} x)\,  \qquad \text{in} \ \  L^2(\R^2).
\end{equation}
Next we consider the quadratic form 
\begin{equation} \label{form-a0} 
\int_{\R^2} \left( |\nabla u|^2 +(\log |x|) \, |u|^2\right)\, dx, \qquad u\in C_0^\infty(\R^2)\, .
\end{equation}
By Lemma \ref{lem-vc} this form is bounded from below. We denote by $q_0$ its closure with the domain:
$$
d(q_0) = \left\{ u\in H^1(\R^2) \, : \, \int_{\R^2} \big| \log |x|\big | \, |u|^2\, dx < \infty\right\}. 
$$
Let $a_0$ be the self-adjoint operator in $L^2(\R^2)$ generated by $q_0$. Then $a_0$ acts on its domain as
\begin{equation}
a_0 = -\Delta + \log |x| \, ,
\end{equation} 
and the spectrum of $a_0$ is purely discrete because the potential is confining. Let
\begin{equation} \label{sp-a0}
E_1 < E_2 < E_3 < \dots
\end{equation}
be the distinguished eigenvalues of $a_0$ (possibly degenerate, with the exception of $E_1$). As for the operator $a_\kappa$, we notice that $\sigma_{\rm es}(a_\kappa) = [0,\infty)$ and that in view of the 
negativity of $V_{\rm ctr}$ the discrete spectrum of $a_\kappa$ is non-empty for all $\kappa$.  We denote 
$$
\mathcal{E}_1(\kappa) : =  \inf \sigma(a_\kappa)  
$$
the lowest eigenvalue of $a_\kappa$. Let $\phi_1$ and $\varphi_\kappa$ be the normalized eigenfunctions of $a_0$ and $a_\kappa$ respectively:
\begin{equation} \label{phi-k}
a_0\,  \phi_1 = E_1\, \phi_1, \qquad a_\kappa\, \varphi_\kappa = \E_1(\kappa)\, \varphi_\kappa, \quad \|\phi_1\|_{L^2(\R^2)}= \|\varphi_\kappa\|_{L^2(\R^2)} =1 . 
\end{equation} 
 
\smallskip

\begin{lem} \label{lem2d-1}
For $\kappa$ large enough it holds 
\begin{equation} \label{simple} 
\sigma(a_\kappa) \cap \Big(-\infty, -w(0) - \log \sqrt{\kappa} + \frac{E_2+E_1}{2} \Big) = \{\E_1(\kappa)\}.
\end{equation}
Moreover, we have 
\begin{equation} \label{lim-e1} 
\lim_{\kappa\to \infty} \Big(\E_1(\kappa) +  \log  \sqrt{\kappa} \Big) = E_1- w(0), 
\end{equation}
and 
\begin{equation} \label{lim-e2} 
\lim_{\kappa\to \infty} \|\varphi_\kappa -\phi_1\|_{L^2(\R^2)} = 0.
\end{equation}
\end{lem}

\begin{proof}
Keeping in mind \eqref{sp-a0} we introduce the operators 
\begin{equation} \label{op-hat}
\hat a_0 = a_0 -E_3 \quad \text{and} \quad \hat a_\kappa = a_\kappa + w(0) +  \log  \sqrt{\kappa}\ -E_3 \, , \quad \kappa \geq \kappa_0:= e^{2 E_3} . 
\end{equation}
Then 
$$
W_\kappa(x) :=  \hat a_\kappa - \hat a_0  = w(0)- w(\kappa^{-\frac 12} |x|) -\log(1+ \kappa^{-\frac 12} |x|). 
$$
Let $u\in\ L^2(\R^2)$ and let $f = (a_0+i)^{-1} u$. Then by the resolvent equation 
\begin{equation} \label{strong}
\|(\hat a_\kappa+i)^{-1}\, u -(a_0+i)^{-1}\, u\| \, \leq \|\, W_\kappa \, f\|, 
\end{equation}
Since $\log (1+|x| )\, f\in L^2(\R^2)$ and $W_\kappa \to 0$ uniformly on compact sets in $\R^2$, it follows that $ \| W_\kappa \, f\| \to 0$ a $\kappa\to+\infty$. Hence $\hat a_\kappa$ converges to $\hat a_0$ in the sense of strong-resolvent convergence as $\kappa\to \infty$. 
On the other hand, in view of \eqref{op-hat} 
\begin{align*}
\hat a_\kappa & = -\Delta + \log\frac{|x|}{1+ \kappa^{-1/2} |x|} -E_3 + w(0)- w(\kappa^{-\frac 12} |x|) \\
& \geq -\Delta + \log\frac{|x|}{1+ \kappa_0^{-1/2} |x|} -E_3 =: S.
\end{align*}
The operator $S$ is bounded from below in $L^2(\R^2)$ and its essential spectrum coincides with the half-line $[0,\infty)$. We can thus apply the result of \cite{weid}. The latter states that the negative eigenvalues of $\, \hat a_\kappa$ converge (including multiplicities) to the negative eigenvalues of $\hat a_0$ as $\kappa\to +\infty$. 
Since $\hat a_0$ has exactly two negative eigenvalues: $E_1-E_3$ and $E_2-E_3$, this implies that
$$
\lim_{\kappa\to \infty} \Big(\E_j(\kappa) +  \log  \sqrt{\kappa}\,  \Big) = E_j- w(0) , \qquad j=1,2,
$$
where $\E_2(\kappa)$ is the second eigenvalue of $a_\kappa$.
Hence \eqref{lim-e1} and \eqref{simple}. Moreover,  the eigenfunctions of $\hat a_\kappa$ relative to negative eigenvalues converge in norm to the eigenfunctions of $\hat a_0$ relative to its negative eigenvalues, see \cite{weid}. As the eigenfunctions of $\hat a_\kappa$ coincide with the eigenfunctions of $a_\kappa$, and the eigenfunctions of $\hat a_0$ coincide with those of $a_0$, we obtain \eqref{lim-e2}. 
\end{proof}


\vspace{0.2cm}

\subsection{Large coupling: absence of discrete spectrum} In this section we prove the absence of discrete spectrum of the operator $\hh_\kappa$ for large $\kappa$. We need some preliminary results.

\begin{lem} \label{lem-vc}
Let $	\kappa \geq 1$. Then for every $\eps>0$ there exists $C_\eps$ independent of $\kappa$ and such that 
\begin{equation}
 \big | \lef  V_{\rm ctr}(\kappa^{-\frac 12} x)  \, u , u \rig_{L^2(\R^2)} \big |  \, \leq \, \left(C_\eps+\log\sqrt{\kappa}\right)\, \|u\|_2^2 + \eps \|\nabla u\|_2^2
\end{equation}
holds for all $u\in H^1(\R^2)$. 
\end{lem}

\begin{proof} 
Let $u\in H^1(\R^2)$. Since $V_{\rm ctr} <0$, $w$ is decreasing and $\kappa\geq 1$, we have 
\begin{align}
\big | \lef  V_{\rm ctr}(\kappa^{-\frac 12} x)  \, u , u \rig \big | - w(0) & \, \leq  \int_{\R^2} \log \Big(1+\frac{\sqrt{\kappa}}{|x|} \Big)\, u^2(x)\, dx  \nonumber \\
&\leq  \log\sqrt{\kappa} \, \|u\|_2^2 +  \int_{\R^2} \log \Big(1+\frac{1}{|x|} \Big)\, u^2(x)\, dx  \nonumber \\
& \leq   \log(2\sqrt{\kappa})\, \|u\|_2^2 +  \int_{\B_1} \log \Big(1+\frac{1}{|x|} \Big)\, u^2(x)\, dx, \label{ball}
\end{align}
where $\B_1 = \{x\in\R^2\, :\, |x| <1\}$. From the compactness of the imbedding $H^1(\B_1) \hookrightarrow L^q(\B_1)$ with $2\leq q<\infty$ it follows that for any $\eps>0$ there exists $C'_\eps$ such that 
\begin{equation} \label{eq-comp}
\left(\int_{\B_1} |u|^q \, dx \right)^{\frac 2q} \ \leq \ \eps \|\nabla u\|_2^2 + C'_\eps \|u\|_2^2 .
\end{equation} 
Since $ \log \big(1+\frac{1}{|x|} \big) \in L^p(\B_1)$ for all $1\leq p< \infty$, the claim follows by an application of the H\"older inequality to the last term in \eqref{ball}. 
\end{proof}

\begin{lem} \label{lem-exp}
Let $\varphi_\kappa$ be given by \eqref{phi-k}. Then there exist $\alpha>0$ and $\kappa_3\geq 1$ such that 
\begin{align} \label{exp-phi}
\sup_{\kappa\geq \kappa_3} \int_{\R^2} e^{2\alpha\sqrt{1+|x|^2}}  \, \varphi^2_\kappa(x)\,   dx  & < \infty \\
\sup_{\kappa\geq \kappa_3}\, \sup_{x\in\R^2}\, (1+|x|^3)\, \varphi_\kappa^2(x) \, &  < \infty. \label{pointwise-ub}
\end{align}
\end{lem}

\begin{proof}
In order to prove \eqref{exp-phi}, we proceed as in the proof of Lemma \ref{lem-2}. 
By Lemma \ref{lem2d-1} there exist $\delta>0$ and $\kappa_\delta$ such that 
\begin{equation} \label{supsup2}
\sup_{z\in\gamma}\,  \sup_{\kappa\geq \kappa_\delta}\, \|(a_\kappa -z )^{-1} \|\  < \ \infty ,
\end{equation}
where
\begin{equation}
\gamma = \{ z\in\C: \ |z-E_1 +w(0)+ \log\sqrt{\kappa}\, | = \delta  \}. 
\end{equation}
Next, for any $z\in\C$ it holds
$$
e^{\alpha\sqrt{1+|\cdot|^2}}\, (a_\kappa-z) \, e^{-\alpha\sqrt{1+|\cdot|^2}}\, = a_\kappa -z +\widehat W_\alpha, 
$$
where $\widehat W_\alpha$ is a first order differential operator which acts the polar coordinates as 
\begin{equation}
\widehat W_\alpha = \frac{2\alpha r}{\sqrt{1+r^2}}\ \pd_r + \frac{\alpha}{(1+r^2)^{3/2}} + \frac{\alpha }{1+r^2}- \frac{\alpha^2 r^2}{1+r^2}\, .
\end{equation} 
For any $u\in  H^1(\R)$ and any $\alpha\in (0,1)$ we then have
\begin{equation}\label{What-upperb}
\|\widehat W_\alpha \, u\|^2 \ \lesssim \  \alpha^2\, ( \|\nabla u\|^2 +  \|u\|^2\, ).
\end{equation}
Now we note that
$$
a_\kappa + \log\sqrt{\kappa} \, \geq \, -\frac 12\, \Delta -C 
$$
holds in the sense of quadratic forms on $H^1(\R^2)$ for all $\kappa \geq 1$ and some $C>0$ independent of $\kappa$, see Lemma \ref{lem-vc}. Therefore 
\begin{align*}
\big | \lef (a_\kappa-z) u, u \rig_{L^2(\R)} \big | &\ \geq \ {\rm Re\, } \lef (a_\kappa-z) u, u \rig_{L^2(\R)} \ \geq\  \frac 12\|\nabla u\|^2 -\widehat C\ \|u\|^2.
\end{align*}
holds true for all $\kappa\geq 1$, all $z\in\gamma$ and some constant $\widehat C>0$ independent of $\kappa$. 
This in combination with \eqref{What-upperb} gives
\begin{align*}
\|\widehat W_\alpha\, u\|^2  & \ \lesssim \ \alpha^2 \big ( \|u\|^2 +  \|(a_\kappa-z) u\|^2\, \big)
\end{align*}
for all $z\in\gamma$. As in the proof of Lemma \ref{lem-2} we conclude that there exists $\alpha \in(0,1)$ such that the operator 
$$
e^{\alpha\sqrt{1+|\cdot|^2}}\, (a_\kappa-z) \, e^{-\alpha\sqrt{1+|\cdot|^2}}\,  = a_\kappa + \widehat W_\alpha-z 
$$
is invertible for all $z\in\gamma$ and all $\kappa\geq\kappa_\delta$, with a bounded inverse, see  \cite[Thm.~IV.1.16]{ka}. In view of the identity
$$
e^{\alpha\sqrt{1+|\cdot|^2}}\,\, (a_\kappa-z)^{-1} \, e^{-\alpha\sqrt{1+|\cdot|^2}}=  (a_\kappa + \widehat W_\alpha-z)^{-1} \, ,
$$
it follows that 
\begin{equation} 
\sup_{z\in\gamma}\,  \sup_{\kappa\geq \kappa_\delta}\, \big\| e^{\alpha\sqrt{1+|\cdot|^2}}\, (a_\kappa -z )^{-1}\, e^{-\alpha\sqrt{1+|\cdot|^2}}\ \big\|\ < \infty. 
\end{equation} 
Now let
\begin{equation} 
\widehat P_\kappa = \varphi_\kappa \lef\, \cdot\, , \varphi_\kappa \rig_{L^2(\R)} .
\end{equation}
Then by Lemma \ref{lem2d-1} and equation \eqref{supsup2} 
\begin{align*} 
\sup_{\kappa\geq \kappa_\delta}\,  \big\| e^{\alpha\sqrt{1+|\cdot|^2}}\, \widehat P_\kappa\, e^{-\alpha\sqrt{1+|\cdot|^2}}\ \big\| & \ \leq \ \sup_{\kappa\geq \kappa_\delta}\, \frac{1}{2\pi} \oint_\gamma  \big\| e^{\alpha\sqrt{1+|\cdot|^2}}\, (a_\kappa -z )^{-1}\, e^{-\alpha\sqrt{1+|\cdot|^2}}\ \big\|\, dz \ < \infty. 
\end{align*}
Since $\varphi_\kappa$ converges strongly to $\phi_0$ in $L^2(\R^2)$ as $\kappa\to\infty$, see \eqref{lim-e2},  we can now follow line by line the arguments of the proof of Lemma \ref{lem-1} and conclude that \eqref{exp-phi} holds true with some $\kappa_3 \geq \kappa_\delta$. 

\medskip

\noindent It remains to prove \eqref{pointwise-ub}. By \eqref{phi-k}
$$
-\Delta \varphi_\kappa = (\E_1(\kappa) +\log  \sqrt{\kappa} + w(\kappa^{-\frac 12} |x|) \big)\,  \varphi_\kappa + \log \Big(\frac{1}{\sqrt{\kappa}} +\frac{1}{|x|}\Big) \,  \varphi_\kappa.
$$
Since $\E_1(\kappa) +\log \sqrt{\kappa} + w(\kappa^{-\frac 12} |x|)$ is bounded in $\R^2$, uniformly with respect to $\kappa$, see \eqref{lim-e1}, it follows that 
$$
\| \Delta\,  \varphi_\kappa\|_2^2 \, \leq \, C +  \int_{\R^2} \log^2 \Big(1+\frac{1}{|x|} \Big)\,  \varphi_\kappa^2(x)\, dx \, \leq  C +\log^22 + \int_{\B_1} \log^2 \Big(1+\frac{1}{|x|} \Big)\,  \varphi_\kappa^2(x)\, dx. 
$$
Now we proceed as in the proof of Lemma \ref{lem-vc}: using \eqref{eq-comp} and the fact that 
\begin{equation} \label{phi-h1}
\|\nabla \varphi_\kappa\|_2^2\, \leq \, \|\Delta \varphi_\kappa\|_2\, \|\varphi_\kappa\|_2 \ \leq \delta \|\Delta \varphi_\kappa\|_2^2 + \delta^{-1} \qquad \forall\ \delta >0, 
\end{equation} 
which follows from integration by parts, we find that  $\| \Delta\,  \varphi_\kappa\|_2$ is bounded uniformly in $\kappa$. The continuity of the Sobolev imbedding  $H^2(\R^2)  \hookrightarrow L^\infty(\R^2)$ then implies that 
\begin{equation} \label{l-infty}
\sup_{\kappa\geq 1} \| \varphi_\kappa\|_\infty \ \lesssim \  \sup_{\kappa\geq 1}   \| \varphi_\kappa\|_{H^2(\R^2)}  < \infty . 
\end{equation} 
On the other hand, since $\varphi_\kappa$ is radial, being the ground-state of a Schr\"odigner operator with a radial potential, an integration by parts in combination with \eqref{l-infty} shows that
\begin{equation} \label{eq-pp}
2 \int_0^r t^3\, \varphi'_\kappa(t)\, \varphi_\kappa(t)\, dt = r^3 \varphi^2_\kappa(r) - 6 \int_0^r t^2\,  \varphi^2_\kappa(t)\, dt. 
\end{equation}
By \eqref{exp-phi}
\begin{equation} 
\sup_{\kappa\geq \kappa_3} \int_0^\infty  r^n\, \varphi^2_\kappa(r)\,  dr < \infty \qquad \forall\, n\geq 1. 
\end{equation}
Hence the claim follows from the H\"older inequality and equations \eqref{phi-h1}-\eqref{eq-pp}.  
\end{proof}

\begin{lem} \label{lem-conv}
There exists a constant $C>0$ such that 
\begin{equation} 
\int_{\R^2} |V_{\rm ctr}(\kappa^{-\frac 12}|x-y|) |\, \varphi_\kappa^2(x)\, dx \ \leq  - C\, V_{\rm ctr}(\kappa^{-\frac 12} y) \qquad \ \forall\ y\in \R^2 \, ,
\end{equation}
holds true for all $\kappa$ large enough. 
\end{lem}

\begin{proof} 
Note that $V_{\rm ctr} <0$ and that 
\begin{align*}
|V_{\rm ctr}(\kappa^{-\frac 12}|x-y|) | & = \log \Big(1+\frac{\sqrt{\kappa}}{|x-y|}\Big) + w(\kappa^{-\frac 12}|x-y|), \\   
-V_{\rm ctr}(\kappa^{-\frac 12} y)  & = \log \Big(1+\frac{\sqrt{\kappa}}{|y|}\, \Big) + w(\kappa^{-\frac 12}|y|)\, .
\end{align*}
Moreover, from the inequality 
\begin{equation} \label{aux-log}
\log \Big(1+\frac{\beta}{t}\, \Big) \ \geq \frac{\beta}{\beta+t} \qquad \forall\  t>0,\ \forall\ \beta>0 ,
\end{equation}
and from the assumptions on $w$ it follows that 
\begin{align*}
w(\kappa^{-\frac 12}|x-y|)\  \lesssim \ \frac{1}{1+\kappa^{-\frac 12}|x-y|} \ \leq \ \log \Big(1+\frac{\sqrt{\kappa}}{|x-y|}\Big) \, .
\end{align*}
Hence in view of the positivity of $w$ to prove the claim it suffices to show that 
\begin{equation} \label{enough-log}
\int_{\R^2}   \log \Big(1+\frac{\sqrt{\kappa}}{|x-y|}\Big) \, \varphi_\kappa^2(x)\, dx \ \leq \ c\,  \log \Big(1+\frac{\sqrt{\kappa}}{|y|}\Big) \qquad \ \forall\ y\in \R^2 .
\end{equation} 
holds for all $\kappa$ large enough and some $c>0$. To simplify the notation we write $t=|x|$ and $r=|y|$ keeping in mind that $\varphi_\kappa$ is radial. Then by Lemma \ref{lem-exp}
\begin{align}
& \int_{\R^2}   \log \Big(1+\frac{\sqrt{\kappa}}{|x-y|}\Big) \, \varphi_\kappa^2(x)\, dx  \  \leq \ 2\pi \int_0^\infty   \log \Big(1+\frac{\sqrt{\kappa}}{|t-r|}\Big) \, \varphi_\kappa^2(t)\, t\,  dt \nonumber  \\
& \qquad = 2\pi  \int_{-r}^\infty   \log \Big(1+\frac{\sqrt{\kappa}}{|t|}\Big) \, \varphi_\kappa^2(r+t)\, (r+t)\, dt \,  \nonumber \\
& \qquad \leq\, 2 \pi \log \Big(1+\frac{2\sqrt{\kappa}}{r}\Big)
 + \int_{-r/2}^{r/2}   \log \Big(1+\frac{\sqrt{\kappa}}{|t|}\Big) \, \varphi_\kappa^2(r+t)\, (r+t)\, dt \nonumber  \\
 & \qquad \leq 4 \pi \log \Big(1+\frac{\sqrt{\kappa}}{r}\Big)
 +  2 \sup_{|t| \leq r/2} \varphi_\kappa^2(r+t)\, (r+t)  \int_{0}^{r/2}   \log \Big(1+\frac{\sqrt{\kappa}}{t}\Big) \, dt \nonumber \\
 & \qquad \leq 4 \pi \log \Big(1+\frac{\sqrt{\kappa}}{r}\Big)
 +  \frac{c\, r}{1+r^3} \int_{0}^{r}   \log \Big(1+\frac{\sqrt{\kappa}}{t}\Big) \, dt ,  \label{eq-conv}
\end{align}
holds for some $c>0$. Here we have used the fact that $\log(1+2x) \leq 2 \log (1+x)$ holds for any $x>0$. A simple calculation shows that
\begin{align*}
 \int_{0}^{r}   \log \Big(1+\frac{\sqrt{\kappa}}{t}\Big) \, dt  & = r \log \Big(1+\frac{\sqrt{\kappa}}{r}\Big) + \sqrt{\kappa}\ \log \Big(1+\frac{r}{\sqrt{\kappa}}\Big) \\
 & \leq  r \log \Big(1+\frac{\sqrt{\kappa}}{r}\Big) + r. 
\end{align*}
This in combination with \eqref{aux-log} and \eqref{eq-conv} proves \eqref{enough-log} and hence the claim. 
\end{proof} 

\vspace{0.2cm}

\subsection{Proof of Theorem \ref{thm-2d}{\rm (i)}}
We are going to prove the absence of discrete spectrum of the operator $A_\kappa$ defined in \eqref{A-2d}. Proposition \ref{prop-es} shows that for $\kappa\geq 1$ it holds 
\begin{equation}
\inf\sigma_{\rm es}(A_\kappa)  = \E_1(\kappa). 
\end{equation}
Since the form domain of $A_\kappa$ coincides with $H^1(\R^4)$, see Lemma \ref{lem-vc}, it suffices to show that 
\begin{equation} \label{enough-ak}
\lef \,  A_\kappa\,  u,\, u\rig_{L^2(\R^4)} \ \geq\  \E_1(\kappa)  \, \|u\|_2^2 \qquad \forall\ u\in H^1(\R^4)
\end{equation}  
holds true for $\kappa$ large enough. Given $u\in  H^1(\R^4)$ we write 
\begin{equation} \label{decomp}
u(x,y) = \varphi_\kappa(x)\, \psi(y) + f(x,y), \quad \psi(y) = \int_{\R^2}  \varphi_\kappa(x)\, u(x,y)\, dx .
\end{equation}
Then 
\begin{equation} \label{orth}
\int_{\R^2} \varphi_\kappa(x)\, f(x,y)\, dx = 0 \qquad \forall\ y\in\R^2,
\end{equation} 
and integrating by parts we obtain
\begin{align}
\lef \,  A_\kappa\,  u,\, u\rig_{L^2(\R^4)} & = \int_{\R^2} |\nabla_y \psi|^2\, dy + \E_1(\kappa) \int_{\R^2} | \psi|^2\, dy + \int_{\R^4} \big( |\nabla f|^2  + V_{\rm ctr}(\kappa^{-\frac 12} x)\, f^2\big)\, dx dy \nonumber\\
& \quad-  \int_{\R^4} \big( V_{\rm ctr}(\kappa^{-\frac 12} y) - \kappa^{-1} V_{\rm ctr}(\kappa^{-\frac 12} |x-y|)\big)\, (\varphi_\kappa \psi +f)^2\, dx dy \nonumber\\
& =  \int_{\R^2} |\nabla_y \psi|^2\, dy + \E_1(\kappa) \, \|u\|_2^2 +  \int_{\R^4} \Big[ |\nabla f|^2 + (V_{\rm ctr}(\kappa^{-\frac 12} x) - \E_1(\kappa)) f^2 \Big]\, dx dy\nonumber\\
& \quad  + \int_{\R^2} \psi^2(y) \left( \int_{\R^2} \kappa^{-1} V_{\rm ctr}(\kappa^{-\frac 12}|x-y|) \, \varphi_\kappa^2(x)\, dx-V_{\rm ctr}(\kappa^{-\frac 12} |y|) \right) dy\nonumber \\
& \quad +  2  \kappa^{-1} \int_{\R^4} V_{\rm ctr}(\kappa^{-\frac 12} |x-y|)\, \varphi_\kappa(x) \psi(y) \, f(x,y) \, dx dy \nonumber\\ 
& \quad - \int_{\R^4} \big( V_{\rm ctr}(\kappa^{-\frac 12} y) - \kappa^{-1} V_{\rm ctr}(\kappa^{-\frac 12} |x-y|)\big)\, f^2\, dx dy \label{ak-lowerb0} .
\end{align}
Hence
\begin{align}
\lef \,  A_\kappa\,  u,\, u\rig_{L^2(\R^4)} & \geq    \E_1(\kappa) \, \|u\|_2^2  \nonumber \\
& \quad +  \int_{\R^4} \Big[ |\nabla f|^2 + (V_{\rm ctr}(\kappa^{-\frac 12} x) - \E_1(\kappa)+ 2\kappa^{-1} V_{\rm ctr}(\kappa^{-\frac 12} |x-y|)) f^2 \Big]\, dx dy \nonumber \\
& \quad + \int_{\R^2} \psi^2(y) \left( \int_{\R^2} 2 \kappa^{-1}V_{\rm ctr}(\kappa^{-\frac 12}|x-y|) \, \varphi_\kappa^2(x)\, dx-V_{\rm ctr}(\kappa^{-\frac 12} |y|) \right) dy, \label{ak-lowerb1}
\end{align}
where we have used the inequality
$$
2 V_{\rm ctr}(\kappa^{-\frac 12} |x-y|)\, \varphi_\kappa \psi \, f \ \geq \ V_{\rm ctr}(\kappa^{-\frac 12} |x-y|)\, \varphi^2_\kappa \psi^2 +  V_{\rm ctr}(\kappa^{-\frac 12} |x-y|)\, f^2,
$$
and  the fact that $V_{\rm ctr} <0$. Note that the last term in \eqref{ak-lowerb1} is positive for $\kappa$ large enough by Lemma \ref{lem-conv}. Moreover, since $\E_1(\kappa)$ is simple,  Lemma \ref{lem2d-1} and equation \eqref{orth} ensure that
\begin{align}  \label{eq-gap} 
 \int_{\R^2} \Big[ |\nabla_x f|^2 + (V_{\rm ctr}(\kappa^{-\frac 12} x) - \E_1(\kappa)\,  f^2(x,y) \Big]\, dx  & \geq \  (\E_2(\kappa)- \E_1(\kappa)) \int_{\R^2} f^2(x,y)\, dx \nonumber \\
 & \geq \frac{E_2-E_1}{2}\,  \int_{\R^2} f^2(x,y)\, dx
\end{align} 
holds for all $y\in\R^2$ and $\kappa$ large enough. Note also that $E_2-E_1>0$ and that $E_1$ is simple. Hence for every $\eta\in (0,1)$  it holds
\begin{align}
\lef \,  A_\kappa\,   u,\, u\rig_{L^2(\R^4)}  -\E_1(\kappa) \, \|u\|_2^2 &\  \geq\   \frac{E_2-E_1}{2} \, \|f\|_2^2 \nonumber \\
& \quad + \frac 12 \int_{\R^4} \Big[ (1-\eta)\, |\nabla_x f|^2 + V_{\rm ctr}(\kappa^{-\frac 12} x)\, f^2\,  -\E_1(\kappa) f^2 \Big]\, dx dy\nonumber \\
& \quad + \int_{\R^4} \Big( \eta\, |\nabla f|^2 +2 \kappa^{-1}V_{\rm ctr}(\kappa^{-\frac 12}|x-y|)\, f^2 \Big)\, dx dy \label{lowerb-1}
\end{align}
By scaling and Lemma \ref{lem2d-1} it follows that for $\kappa\to\infty$ 
\begin{equation} 
\inf \sigma \big( -(1-\eta) \Delta_x +V_{\rm ctr}(\kappa^{-\frac 12} x)  -\E_1(\kappa)\big) = \log(1-\eta)+ o_\kappa(1)\qquad \forall \ \eta\in (0,1),
\end{equation} 
where $o_\kappa(1)$ denotes a quantity which tends to zero as $\kappa\to\infty$. 
Hence inserting
$$
\eta=1-\exp((E_1-E_2)/4)
$$ 
into \eqref{lowerb-1} we get
\begin{align*}
\lef \,  A_\kappa  u,\, u\rig_{L^2(\R^4)}  -\E_1(\kappa) \, \|u\|_2^2 &\  \geq\  \Big(  \frac{E_2-E_1}{4}+o_\kappa(1) \Big) \, \|f\|_2^2  \\
& \quad + \int_{\R^4} \Big( \eta\, |\nabla f|^2 +2 \kappa^{-1}V_{\rm ctr}(\kappa^{-\frac 12}|x-y|)\, f^2 \Big)\, dx dy
\end{align*}
In order to estimate the second term on the right hand side we use again the change of variables $(x,y)\mapsto (s,t)= (x-y, \frac{x+y}{2})$. This gives 
\begin{align*}
\int_{\R^4} \Big( \eta\, |\nabla f|^2 +2 \kappa^{-1}V_{\rm ctr}(\kappa^{-\frac 12}|x-y|)\, f^2 \Big)\, dx dy& = \int_{\R^4} \Big(2 \eta\, |\nabla_s g|^2 + \frac 12\, \eta\, |\nabla_t g|^2 \\ 
& \qquad \ \quad +2 \kappa^{-1}V_{\rm ctr}(\kappa^{-\frac 12}|s|)\, g^2 \Big)\, ds dt,
\end{align*}
where $g(s,t)= f(t+s/2, t -s/2)$. 
In view of Lemma \ref{lem-vc} we then obtain the lower bound 
\begin{align*}
\lef \,  A_\kappa\,  u,\, u\rig_{L^2(\R^4)}  -\E_1(\kappa) \, \|u\|_2^2 &\  \geq\   \left( \frac{E_2-E_1}{4} -  \mathcal{O}(\kappa^{-1}\log \kappa) +o_\kappa(1) \right) \, \|f\|_2^2,
\end{align*}
which proves \eqref{enough-ak}. 
\qed

\vspace{0.2cm}

\subsection{Proof of Theorem \ref{thm-2d}{\rm (ii)}}

\begin{prop} \label{prop-infty}
If $\kappa \in [k_e, 1)$, then $\hh_\kappa$ has infinitely many discrete eigenvalues. 
\end{prop}
\begin{proof}
Let  $\kappa \in [k_e, 1)$. In view of Proposition \ref{prop-es} and the variational principle it suffices to show that there exists a subspace $\F_\kappa \subset H^1(\R^4)$ such that 
\begin{equation} \label{var-pr}
{\rm dim} (\F_\kappa) = \infty \quad \wedge \quad \forall\, u \in \mathcal F_\kappa :\ \lef \,  A_\kappa\,  u,\, u\rig_{L^2(\R^4)}  < \E_1(\kappa) \, \|u\|_2^2 .
\end{equation} 
By choosing 
$$
u(x,y) = \varphi_\kappa(x)\, \psi(y), \ \psi\in H^1(\R^2)
$$
and using the calculations made in the proof of Theorem \ref{thm-2d} (with $f=0$) we obtain the identity 
\begin{align*}
\lef \,  A_\kappa\,  u,\, u\rig_{L^2(\R^4)}  & =  \E_1(\kappa) \, \|u\|_2^2 + \int_{\R^2} |\nabla_y \psi|^2\, dy + \int_{\R^2} U_\kappa(y)\, \psi^2(y) \, dy ,
\end{align*}
where
$$
U_\kappa(y) =  \int_{\R^2} \kappa^{-1} V_{\rm ctr}(\kappa^{-\frac 12}|x-y|) \, \varphi_\kappa^2(x)\, dx-V_{\rm ctr}(\kappa^{-\frac 12} |y|) .
$$
A direct calculation now shows that 
$$
\lim_{|y|\to \infty}\, |y|\, U_\kappa(y) = \kappa^{1/2}- \kappa^{-1/2} \ <0 . 
$$
Hence by standard results of spectral theory it follows that there exists an infinite-dimensional subspace $\mathcal G_\kappa \subset H^1(\R^2)$ such that 
$$
\int_{\R^2} |\nabla_y \psi|^2\, dy + \int_{\R^2} U_\kappa(y)\, \psi^2(y) \, dy < 0 \qquad \forall\, \psi \in \mathcal G_\kappa. 
$$
Setting
$\F_\kappa = \{ u\in H^1(\R^4): u(x,y) = \varphi_\kappa(x)\, \psi(y), \ \psi\in \mathcal G_\kappa\}$ then completes the proof of \eqref{var-pr}. 
\end{proof}

\subsection{Small coupling} 

\begin{lem} \label{lem-monotone}
The number of discrete eigenvalues of the operator $\HH_\kappa(v)$ is non-decreasing in $\kappa$ on the interval $(0,k_e]$. 
\end{lem}

\begin{proof}
By Proposition \ref{prop-es} the number of discrete eigenvalues of the operator $\HH_\kappa(V)$ is equal to $N(\HH_\kappa(V), \Lambda_0(V))_{L^2(\R^4)}$ for all $\kappa$ in the interval $(0,k_e]$. Let $0 < \kappa < k_e$ and assume that $N(\HH_\kappa(V), \Lambda_0(V))_{L^2(\R^4)} =N \geq 1$. Then there exist $\psi_1, \psi_2 \dots \psi_N \in H^1(\R^4)$ (which can be chosen real valued) and $E_1, E_2, \dots E_N$ such that 
$$
\HH_\kappa(V)\, \psi_j = E_j\, \psi_j, \quad E_j < \Lambda_0(v) \qquad \forall \, j=1,\dots, N,
$$
and $\lef \psi_j, \psi_k \rig_{L^2(\R^4)} = \delta_{jk}$. 
Since by definition of $\Lambda_0(V)$:
$$
\int_{\R^4} |\nabla \psi_j|^2\, dx dy - \int_{\R^4} \psi_j^2(x,y) v(x-y)\, dx dy \, \geq \, \Lambda_0(V) >\ 
E_j=\langle \HH_\kappa(v)\psi_j, \psi_j\rangle \quad \forall\, j=1,\dots, N,
$$
it follows that 
\begin{equation} \label{psij}
\int_{\R^4} \psi_j^2(x,y) v(y)\, dx dy - \int_{\R^4} \psi_j^2(x,y) v(x)\, dx dy < 0 \qquad \forall \, j=1,\dots, N. 
\end{equation}
Now let $ \kappa' \in (\kappa, k_e]$. Then in view of \eqref{psij} we have 
$$
\lef \,  \HH_{\kappa'}(V)\,   \psi_j,\, \psi_j\rig_{L^2(\R^4)}  \, < \, \lef \,  \HH_\kappa(V)\,   \psi_j,\, \psi_j\rig_{L^2(\R^4)}  < \Lambda_0(v) \qquad \forall \, j=1,\dots, N,
$$
and since $\psi_j$ are mutually orthonormal, this implies that $N(\HH_{\kappa'}(v), \Lambda_0(v))_{L^2(\R^4)}\geq N$. 
\end{proof}


\appendix

\section{}
\label{sec-app} 

\begin{lem} \label{lem-app}
Let $v$ satisfy assumption \ref{ass-v} and let $0<\eps < \frac 14$. Then there exists a constant $c_1>0$ such that 
\begin{equation} \label{eq-lowerb}
\Delta_\kappa(x) = \sqrt{\kappa}\, (v(0) -v ( \kappa^{-\frac 14}\, x) ) + \kappa^{2\eps}\, \chi(\kappa^{-\eps} x)\, \geq \, c_1\, \kappa^{2\eps} \qquad \forall\ x\in\R, \  \ \forall\ \kappa\geq 1. 
\end{equation}
\end{lem}

\begin{proof}
Since $\Delta_\kappa(x) =\Delta_\kappa(-x)$, it suffices to prove \eqref{eq-lowerb} for all $x\geq 0$. 
Let $a>0$ such that $v(x) =0$ whenever $x >a$. Let $t := x \kappa^{-1/4}$ and fix a $\kappa\geq 1$. From assumption \ref{ass-v} and the definition of $\Delta_\kappa$ and $\chi$ it follows that 
$$
\Delta_\kappa(x) = \Delta_\kappa(t \, \kappa^{\frac 14}) \ \geq\ \min\{v(0), 1\} \, \kappa^{2\eps} \qquad \forall\ t\in[0,\kappa^{\eps-\frac 14} ) \cup (a,\infty).
$$
Now define 
$$
t_0 := \frac{3\, \omega^2}{\|v'''\|_\infty}\, ,
$$
keeping in mind that by assumption \ref{ass-v} we have $\|v'''\|_\infty >0$. In view of \eqref{taylor} it then follows that 
$$
\sqrt{\kappa}\, v(0) -\sqrt{\kappa}\, v(t)\ \geq\ \sqrt{\kappa}\ \, \frac{\omega^2}{2}\, t^2 \, \geq \ \frac{\omega^2}{2}\, \kappa^{2\eps}\, \qquad \forall\ t\in \big [\kappa^{\eps-\frac 14}\, ,\,  t_0\big ].
$$
To complete the proof we note that the function $v(0) -v(t)$ attains a positive minimum on $[t_0, a]$. Therefore
$$
\sqrt{\kappa}\,v(0) -\sqrt{\kappa}\, v(t) \, \geq \, c\, \sqrt{\kappa} \qquad \forall\ t\in [t_0, a]
$$
holds true with some $c>0$ independent of $\kappa$. 
\end{proof}

\begin{lem} \label{lem-app-2} 
Let $b_2$ and $\kappa_2$ be given by equation \eqref{ab} and Lemma \ref{lem-2} respectivelly. Then for all $\kappa \geq \kappa_2$  it holds
\begin{align}
\int_{\R^2} \psi^2_\kappa(x') \,  v^2(\kappa^{-\frac 14}\, (x'-y')) \, y'^2  \, dx' dy' & \leq \   2 \kappa^{\frac 34}\, m_{2,2} + 2 \kappa^{\frac 14}\, m_{0,2}\, \mu_2\,  \label{aux-4} \\
\int_\R b^2_2(y)\,  dy & \leq \ \kappa^{\frac 14}\, m'_{0,1}\, \mu_1  \label{aux-1}\\
\int_\R b^2_2(y)\, y^2\,  dy  & \leq \  \kappa^{\frac 34}\, m'_{2,1}\, \mu_1 + \kappa^{\frac 14}\, m'_{0,1}\, \mu_3 \label{aux-2}
\end{align}
where we adopted notation \eqref{notation}.
\end{lem}

\begin{proof}
We have
\begin{align*}
\int_{\R^2} \psi^2_\kappa(x') \,  v^2(\kappa^{-\frac 14}\, (x'-y')) \,  y'^2\, dx' dy' & = \kappa^{\frac 14}\, \int_{\R^2} \psi^2_\kappa(x') \, (r+x')^2\,  v^2(r)\, dr dx'  \\
& = 2 \kappa^{\frac 14}\, \int_{\R^2} \psi^2_\kappa(x') \, (r^2+x'^2)\,  v^2(r)\, dr dx'  \\
& = 2 \kappa^{\frac 34}\, m_{2,2} + 2 \kappa^{\frac 14}\, m_{0,2}\, \mu_2\, .
\end{align*}
To prove \eqref{aux-1} we introduce the new variable $s= \kappa^{-\frac 14}\, (y-t)$ to get
\begin{align*}
\int_\R b^2_2(y)\,  dy & \leq \ \int_\R \psi^2_\kappa(x) \int_0^x \Big(\int_\R |v'(s)|\,  \kappa^{\frac 14}\, ds \Big)\, dt dx \ \leq \  \kappa^{\frac 14}\, m'_{0,1}  \int_\R \psi^2_\kappa(x) |x|\, dx  \\
& = \kappa^{\frac 14}\, m'_{0,1}\, \mu_1 \, .
\end{align*}
Similarly, 
\begin{align*}
\int_\R b^2_2(y)\, y^2\,  dy & \leq \ \int_\R \psi^2_\kappa(x) \int_0^x \Big(\int_\R (\kappa^{\frac 14}\, s+t)^2 |v'(s)|\,  \kappa^{\frac 14}\, ds \Big)\, dt dx \\ 
& =  \int_\R \psi^2_\kappa(x) \int_0^x \Big(\int_\R (\kappa^{\frac 12}\, s^2+t^2)\, |v'(s)|\,  \kappa^{\frac 14}\, ds \Big)\, dt dx \\
& \leq \ \kappa^{\frac 14} \int_\R \psi^2_\kappa(x) (\kappa^{\frac 12}\, m'_{2,1}\, |x| + m'_{0,1}\, |x|^3)\, dx \\
& = \kappa^{\frac 34}\, m'_{2,1}\, \mu_1 + \kappa^{\frac 14}\, m'_{0,1}\, \mu_3\, .
\end{align*}
\end{proof}

\vspace{0.2cm}

\noindent{\bf Acknowledgements.} 

\vspace{0.2cm}

T.G.P. is supported by the QUSCOPE Center, which is funded by the Villum Foundation. H.C. was partially supported by the Danish Council of Independent Research | Natural Sciences, Grant No. DFF-4181-00042. H.K. was partially supported by the No. MIUR-PRIN2010-11 grant for the project "Calcolo delle variazioni."



\begin{thebibliography}{99}





\bibitem{cud} P.~Cudazzo, I.V.~Tokatly, A.~Rubio: Dielectric screening in two-dimensional insulators: Implications for excitonic
and impurity states in graphane. {\em Phys.~Rev.~B} {\bf 84} (2011) 085406. 

\bibitem{Ganchev} B.~Ganchev, N.~Drummond, I.~Aleiner, and V.~Fal'ko: Three-Particle Complexes in Two-Dimensional Semiconductors. {\em Phys.~Rev.~Lett.}  {\bf 114} (2015) 107401.

\bibitem{hkpc} J.~Have, H.~Kova\v{r}\'{\i}k, T.G.~Pedersen, and H.~Cornean: On the existence of impurity bound excitons in one-dimensional systems with zero range interactions. {\em J.~Math.~Phys.} {\bf 58} (2017)  052106, 16 pp.
	
\bibitem{ka}  T.~Kato: {\em Perturbation theory for linear operators }. Springer-Verlag, 1980.

\bibitem{Keldysh} L.V.~Keldysh: Coulomb interaction in thin semiconductor and semimetal films. {\em JETP Letters} {\bf 29} (1978) 658.	

\bibitem{Mostaani}
  E.~Mostaani, M.~Szyniszewski, C.~H.~Price, R.~Maezono, M.~Danovich, R.~J.~Hunt, N.~D.~Drummond, and V.~I.~Fal'ko: Diffusion quantum Monte Carlo study of excitonic complexes in two-dimensional transition-metal dichalcogenides. {\em Phys.~Rev.~B} {\bf 96} (2017) 075431.

\bibitem{pra} E.~Prada et al: Effective-mass theory for the anisotropic exciton in two-dimensional crystals:
Application to phosphorene. {\em Phys.~Rev.~B} {\bf 91} (2015) 245421. 


\bibitem{rs4}  M.~Reed, B.~Simon: 
{\em Methods of Modern of Mathematical Physics, IV: Analysis of Operators}. Academic
Press, New York-London, 1978.
	
\bibitem{rpc}	 T.F.~R{\o}nnow, T.G.~Pedersen, and H.~Cornean: Stability of singlet and triplet trions in carbon nanotubes. {\em Phys.~Lett.~A.} {\bf 373} (2009) 1478.

\bibitem{si2} B.~Simon: The Bound State of Weakly Coupled Schr\"odinger Operators in One and Two Dimensions. {\em Ann.~Phys.} {\bf 97} (1976)  279--288. 

\bibitem{si} B.~Simon: Geometric Methods in Multiparticle Quantum Systems. {\em Comm.~Math.~Phys.} {\bf 55} (1977)   259--274.


\bibitem{schur} I. Schur: Neue Begr\"undung der Theorie der Gruppencharaktere. {\it Sitzungsberichte
der K\"oniglich Preu{\ss}ischen Akademie der Wissenschaften zu Berlin} 1905, 406-432.

\bibitem{tpv} M.L.~Trolle, T.G.~Pedersen, and V.~Veniard: Model dielectric function for 2D semiconductors including substrate screening. {\em Sci.~Rep.} {\bf 7} (2017) 39844.

\bibitem{weid} J.Weidmann: Continuity of the eigenvalues of selfadjoint operators with respect to the strong operator topology. {\em Integral Equations Operator Theory} {\bf 3} (1980) 138--142.





\end{thebibliography}
 \end{document}